\documentclass[11pt]{article}
\usepackage{geometry}
\usepackage{amsfonts,amsmath,amsthm,amssymb,cancel}
\usepackage[normalem]{ulem}
\usepackage[colorlinks=true, pdfstartview=FitV, linkcolor=black, citecolor=black, urlcolor=blue]{hyperref}

\usepackage{tikz, mathrsfs}
\usepackage{tikz, pgfplots, todonotes,ulem,amssymb}
\usetikzlibrary{decorations.markings,arrows, calc, shapes, intersections, patterns}

\newcommand {\TriangleCut}[2]{ 
     \foreach [count=\j] \vv in {1,1,1}{  
     \draw[decoration={aspect=0.5, segment length=02, amplitude=0.7,coil},decorate] (0,0) to (-180 +120*\j-120 :1);
     \draw (-180 +120*\j-120 :0.49) 
             to  node [pos=0.6,above] {\small $J$}  (-180 +120*\j-120:0.51);}

\draw [
    postaction={decorate,decoration={markings,mark=at position 0.11 with {\arrow[line width=1pt]{<}}}},
    postaction={decorate,decoration={markings,mark=at position 0.247 with {\arrow[line width=1pt]{>}}}},
    postaction={decorate,decoration={markings,mark=at position 0.44 with {\arrow[line width=1pt]{<}}}},
    postaction={decorate,decoration={markings,mark=at position 0.584 with {\arrow[line width=1pt]{>}}}},
    postaction={decorate,decoration={markings,mark=at position 0.904 with {\arrow[line width=1pt]{>}}}},
    postaction={decorate,decoration={markings,mark=at position 0.77 with {\arrow[line width=1pt]{<}}}},
    ] (0:2) to  (120:2) to  (-120:2) to  cycle ;
       \foreach  [count=\j]  \vv in {1,1,1} {
        \draw [line width=0.4pt, 
        blue,postaction={decorate,decoration={markings,mark=at position 0.55 with {\arrow[line width=1pt]{<}}}}] (0,0) to  node[pos=0.4, above]   
           {\small $L $} (120*\j-120:2);
        \draw [fill](120*\j-120:2) circle[radius=0.02];
      }

}

\newcommand {\Triangleplain}[2]{ 
  
\ifthenelse{#1=0}{
     \foreach [count=\j] \vv in {0,0,0}{ 
          \ifthenelse{\vv=1}
          { \draw [decoration={aspect=0.5, segment length=02, amplitude=0.7,coil},decorate]  (0,0) 
             to  node [pos=0.6,above,sloped] {\ifthenelse{#2=1}{$i\sigma_2$}{}} (-180 +120*\j-120:1);}
          {\draw [decoration={aspect=0.5, segment length=02, amplitude=0.7,coil},decorate] 
              (0,0) to  node [pos=0.6,above,sloped] {\ifthenelse{#2=1}{$i\sigma_2$}{}} (-180 +120*\j-120 :1);}
       }
      }{}
  \ifthenelse{#2=0} {
    \fill  [red!30!white, opacity=0.5 ] (60:1) to (0,0) to (0:2) to cycle;
    \fill [red!30!white, opacity=0.5 ] (-60:1) to (0,0) to (-120:2) to cycle;
    \fill [red!30!white, opacity=0.5 ] (120:2) to (0,0) to (-180:1) to cycle;
  }{}
     \draw [
    postaction={decorate,decoration={markings,mark=at position 0.106 with {\arrow[line width=1.5pt]{<}}}},
    postaction={decorate,decoration={markings,mark=at position 0.247 with {\arrow[line width=1.5pt]{>}}}},
    postaction={decorate,decoration={markings,mark=at position 0.436 with {\arrow[ line width=1.5pt]{<}}}},
    postaction={decorate,decoration={markings,mark=at position 0.5765 with {\arrow[ line width=1.5pt]{>}}}},
    postaction={decorate,decoration={markings,mark=at position 0.91 with {\arrow[line width=1.5pt]{>}}}},
    postaction={decorate,decoration={markings,mark=at position 0.77 with {\arrow[line width=1.5pt]{<}}}},
    ] (0:2) to  (120:2) to  (-120:2) to  cycle ;
       \foreach  [count=\j]  \vv in {1,1,1} {
        \ifthenelse{\vv=1}{
        \draw [line width=0.4pt, 
        blue,postaction={decorate,decoration={markings,mark=at position 0.55 with {\arrow[line width=1.2pt]{<}}}}] (0,0) to  node[pos=0.4, above,sloped]   
           {\ifthenelse{#2=1}{\tiny $L $}{}} (120*\j-120:2);
        \draw [fill](120*\j-120:2) circle[radius=0.02];
        }
}}

\def\Xint#1{\mathchoice
{\XXint\displaystyle\textstyle{#1}}%
{\XXint\textstyle\scriptstyle{#1}}%
{\XXint\scriptstyle\scriptscriptstyle{#1}}%
{\XXint\scriptscriptstyle\scriptscriptstyle{#1}}%
\!\int}
\def\XXint#1#2#3{{\setbox0=\hbox{$#1{#2#3}{\int}$ }
\vcenter{\hbox{$#2#3$ }}\kern-.6\wd0}}

\def \Bil#1#2{\le\langle #1, #2\ri\rangle}
\def\slint{\Xint\setminus}

\def\Fcal{{\mathcal F}}
\def\Gcal{{\mathcal G}}

\def\Qcal{{\mathcal Q}}

\def\e{\hbar}
\def\Qcal{{\mathcal Q}}

\def\Qcal{{\mathcal Q}}
\def \CC{{\mathcal C}}
\newlength{\dinwidth}
\def\le{\left}

\def \QED{\hfill $\blacksquare$\par \vskip 4pt}
\def\ri{\right}

\newlength{\dinmargin}
\def \bea#1\eea {\begin{align} #1 \end{align}}

\def \wt{ \widetilde }
\def \&{\hspace{-15pt}&}
\def \d{\delta}
\def\res{\mathop{ \mathrm {res}}}

\newtheorem{theorem}{Theorem}[section]
\newtheorem{proposition}{Proposition}[section]
\newtheorem{corollary}{Corollary}[section]
\newtheorem{remark}{Remark}[section]
\newtheorem{lemma}{Lemma}[section]

\def\be{\begin{equation}}
\def\ee{\end{equation}}
\def\ben{\begin{displaymath}}
\def\een{\end{displaymath}}
\def\baa{\begin{eqnarray}}
\def\eaa{\end{eqnarray}}

\def\ba{\begin{array}}
\def\ea{\end{array}}
\makeatletter
\@addtoreset{equation}{section}
\makeatother

\def\Acal{\mathcal A}
\def\qd{Q}

\def \eqref #1{(\ref{#1})}
\def \1{\mathbf 1}
\def \br{\begin{remark}}
\def\er{\end{remark}}

\def\vt\tilde{v}

\def\C{{\mathbb C}}
\def\Z{{\mathbb  Z}}
\def\a{\alpha}
\def\g{\gamma}
\def\b{\beta}

\def\l{\lambda}

\def\p{\partial}
\def\Ccal{{\mathcal C}}
\def\Ch{{\widehat{{\mathcal C}}}}

\def\pa{\partial}

\def\f{\frac}
\def\la{\label}
\def\Scal{{\mathcal S}} 

\def\res{\mathop{\mathrm {res}}\limits_}
\def\tr{{\rm tr}}

\def\la{\label}
\def\CC{{\mathcal C}}

 \title{Generating function of monodromy symplectomorphism for Fuchsian systems on ${\bf CP}1$ and its WKB 
 expansion}
\begin{document}
\begin{center}
\huge{Generating function of monodromy symplectomorphism for $2\times 2$  Fuchsian systems  and its WKB 
 expansion}
\end{center}

\begin{center}
\bigskip
M. Bertola$^{\dagger\ddagger\diamondsuit}$\footnote{Marco.Bertola@concordia.ca, mbertola@sissa.it},  
D. Korotkin$^{\dagger\ddagger}$ \footnote{Dmitry.Korotkin@concordia.ca},
F. del Monte$^{\dagger\ddagger}$ \footnote{Fabrizio.Delmonte@concordia.ca}
\\
\bigskip
\begin{small}
$^{\dagger}$ {\it   Department of Mathematics and
Statistics, Concordia University\\ 1455 de Maisonneuve W., Montr\'eal, Qu\'ebec,
Canada H3G 1M8} \\
\smallskip
$^{\ddagger}$ {\it  Centre de recherches math\'ematiques,
Universit\'e de Montr\'eal\\ C.~P.~6128, succ. centre ville, Montr\'eal,
Qu\'ebec, Canada H3C 3J7} \\
\smallskip
$^{\diamondsuit}$ {\it  SISSA/ISAS,  Area of Mathematics\\ via Bonomea 265, 34136 Trieste, Italy }\\
\end{small}
\end{center}
\vspace{0.5cm}

{\bf Abstract.} We study the WKB expansion of $2\times 2$ system of linear differential equations with four
fuchsian singularities. The main focus is on the generating function of the monodromy symplectomorphism
which, according to a recent paper \cite{BK2iso}  is closely related to the Jimbo-Miwa tau-function.
We compute the first three terms of the WKB expansion of the generating function and establish the link to the Bergman tau-function.

\tableofcontents

\section{Introduction}
Although a subject of by now venerable age, the Wentzel-Kramers-Brillouin (WKB) approximation, used since early days of quantum 
mechanics to study the quasi-classical limit of the Schr\"odinger equation, has enjoyed a surge of interest in the past decades, with a positive feedback of results between the mathematical and physical community. After the method was developed further by many mathematicians in the decades preceding the turn of the millennium (see the relatively recent reviews \cite{DDP,kawai2005algebraic}), a new surge of interest in the subject was prompted by the emergence of a connection between the WKB approximation and the geometry of four-dimensional supersymmetric field theories in \cite{GMN}, where the rich geometry arising from the WKB graph and differentials was used to study BPS states of four-dimensional supersymmetric theories. This perspective has been later made more mathematically precise in  \cite{BS,Bri,AlBrid}. The central object in this analysis is the WKB curve, a Riemann surface arising from the leading semiclassical approximation, that coincides with the Seiberg-Witten curve of the associated quantum field theory. Over the curve $\CC$ one introduces the graph of horizontal trajectories for the projective connection entering as potential in the Schr\"odinger equation, defining a triangulation that allows to relate periods of the WKB differential to Fock-Goncharov coordinates \cite{FG,BK_JDG}, defined on the monodromy manifold of the second order ODE.

The monodromy manifold of the Schr\"odinger equation is the $SL_2(\mathbb{C})$ character variety, that can be parametrized by Fock-Goncharov coordinates (in turn related to WKB periods). Over this space the Goldman Poisson bracket \cite{Gold84} is defined, inverted by the symplectic form $\Omega_G$ found in \cite{AlekMal2} on symplectic leaves $V^{{\bf r}}$, which was computed in \cite{BK_JDG} using complex shear (Fock-Goncharov) coordinates. In the paper \cite{BK_TMF} a natural set of Darboux coordinates for the Goldman symplectic form $\Omega_G$, called homological shear coordinates, were found. The symplectic leaf $V^{{\bf r}}$ of the $SL(2)$ character variety, where $e^{\pm2\pi i r_j}$ are the monodromy eigenvalues at the punctures, is the image under the monodromy map $\mathcal{F}$ of the moduli space of meromorphic flat connections with fixed residues, $\Acal^{{\bf r}}$. This is also a symplectic manifold, endowed with the Atiyah-Bott (pre-)symplectic form, that reduces to the Kirillov-Konstant symplectic form $\Omega_{KK}$ in genus 0. According to the theorem proved in \cite{Hitchin,KorSam,AlekMal2}, the monodromy map for a Fuchsian differential equation  is a symplectomorphism between the two spaces  $V^{{\bf r}}$ and $\Acal^{{\bf r}}$,  i.e.
\begin{equation}
\mathcal{F}^*\Omega_G=2\pi i\Omega_{KK}.
\end{equation}

Another connection between monodromy  of linear ODEs and supersymmetric QFTs comes from the theory of tau functions of isomonodromic systems, first introduced in the '80s by the Japanese School \cite{JMU1}. Starting from \cite{GIL2012}, the tau function of a large class of isomonodromy problems, including in particular the sixth Painlev\'e equation and the Schlesinger system \cite{ILT2015,GL2018} as cases relevant to this work, was identified with a Fourier series of non-perturbative four-dimensional gauge theory partition functions. The  quantum field theory  corresponding to a given isomonodromic problem can be identified by the singularity structure of the linear system, or equivalently by its spectral curve, which coincides with the WKB curve \cite{BLMST2017}. In this context, the tau function is defined by requiring its logarithmic derivatives to be the isomonodromic Hamiltonian, which makes it determined  up to an arbitrary function of the monodromy data.

In \cite{BK2iso}, after previous results in this direction in the papers \cite{ItLTy2014,ILP}, it was shown that it is possible to extend the definition of the isomonodromic tau function for Fuchsian systems on the Riemann sphere in a way that fixes not only the time dependence, but also the dependence on the monodromy parameters, by defining it to be the generating function for the monodromy symplectomorphism, that is, given symplectic potentials $\theta_{KK}$ and $\theta_G$ such that $\delta\theta_{KK}=\Omega_{KK}$, $\delta\theta_G=\Omega_G$, the tau function $\mathcal{T}$ is defined as
\begin{equation}\label{eq:dlogT}
d\log\mathcal{T}:=  \theta_{KK}-\sum_{k=1}^{n}H_k \,dt_k {-\frac 1{2i\pi}} \Fcal^* \theta_G.
\end{equation}
In \cite{DMDG2022} it was shown that this same expression, together with its generalization for Fuchsian systems over genus one Riemann Surfaces, arises from the Fredholm determinant representation of the tau function \cite{GL2018,DMDGG2020}. The definition \eqref{eq:dlogT} has the conceptual advantage of fully fixing the functional dependence of $\mathcal{T}$ on all local coordinates of $\Acal^{{\bf r}}$. In particular, it allows to determine the ratio of tau functions expressed in different monodromy coordinates, allowing to determine the so-called connection constant \cite{Its2016,ILP}. In terms of 2d CFT, the connection constant is interpreted as the ratio between conformal blocks in different channels, known as the fusion kernel \cite{ILTy2013,GMS2020,BGG2021}, and in terms of the corresponding 4d gauge theory \cite{Nekrasov:2020qcq,Jeong:2020uxz} it is the ratio of dual partition functions defined in different gauge theory regimes\footnote{A different point of view was taken in \cite{CLT}, where tau functions were characterized as \textit{difference} generating functions. While there are clear similarities between the two definitions, the precise relation is not yet completely clear, and we leave it to future studies.}. 

In this paper we take the first step towards the WKB analysis of the generating function of monodromy symplectomorphism, defined by
\begin{equation}\label{eq:Gdef}
\delta\log\Gcal:=\Fcal^* \theta_G-2\pi i\theta_{KK},
\end{equation}
for the case of a general Fuchsian system on the Riemann sphere, by computing explicitly its first three contributions. From a WKB standpoint, the main difference with respect to previous works such as \cite{GMN,BK_TMF} is that   the $2\times 2$ Schlesinger system gives rise to the  Schr\"odinger equation with apparent singularities. Furthermore, we consider variations $\delta$ that do not move the position of the (non-apparent) singularities, that we denote by $z_1,\dots,z_{g+2}$. 
Even though the apparent singularities  introduce technical complications, it is still possible to explicitly integrate the equation for the generating function in the first three orders, resulting in Theorems \ref{thm:Gm2}, \ref{thm:Gm1}, \ref{thm:G0}.

To obtain the WKB-expansion of the  {\it  isomonodromic} tau function from our computation, one would have to consider also variation of the positions of the poles, and then impose that the apparent singularities evolve according to the isomonodromic deformation equations, as it was  done in \cite{Iwaki:2019zeq,BM} for the case of Painlev\'e I. In this case $R_j$'s become $\hbar$ and time-dependent, so to get the true asymptotic expansion in $\hbar$, one would have to further expand the resulting expression imposing the isomonodromic time evolution.

Let us now introduce the necessary definitions and notations (for more details about the notations we refer the reader to the beginning of Section \ref{sec3}).   We are going to  study the $\hbar$-expansion for the equation

\be
\label{int1}
 \frac{d\Psi}{d z}= R(z)\Psi(z)= \frac{1}{\hbar}\sum_{j=1}^{g+2} \frac {R_j}{z-z_j}\Psi(z).
\ee

Let
\be
R_\infty = R_{g+3}=-\sum_{j=1}^{g+2} R_j\;,
\ee
and assume that $R_\infty$ is diagonal,
\be
R_\infty=\left(\ba{cc} r_\infty & 0 \\ 0 & -r_\infty\ea\right).
\ee
Denote the eigenvalues of the matrices $R_j$ by $r_j$ and $-r_j$, $j=1,\dots,g+3$.
The solution $\Psi$ of (\ref{int1}) has monodromies $M_1,\dots,M_{g+3}$ around the points 
$z_1,z_2,\dots,z_{g+2},\infty$ which satisfy
the relation
$$
M_1\dots M_{g+3}=I.
$$
Assume that the matrices $R_j$ are diagonalizable, and let
\be
R_j= G_j L_j G_j^{-1}
\la{Adia}
\ee
where $L_j={\rm diag}(r_j, -r_j)$. Then, on the space $\Acal^{{\bf r}}$ which is the symplectic leaf $r_j=const$ quotiented over simultaneous transformations of the form
$R_j\to G R_j G^{-1}$, $G\in SL(2,\mathbb{C})$, the Kirillov-Kostant symplectic form  is defined by  
\be
\Omega_{KK}=-\frac{1}{\hbar}\tr \sum_{k=1}^g L_k G_k^{-1}\delta G_k \wedge G_k^{-1}\delta G_k
\la{sf}\ee
 with the symplectic potential $\theta_{KK}$ (such that $\d \theta_{KK}=\omega_{KK}$) given by
\be
\theta_{KK}=\frac{1}{\hbar}\tr\sum_{k=1}^g L_k G_k^{-1}\delta G_k .
\la{thetaKK}
\ee
The construction of the Darboux homological shear coordinates for the Goldman symplectic form \cite{BK_TMF}, parametrizing the monodromy representation of solutions of \eqref{int1}, looks as follows:
write the coefficient matrix of the linear system \eqref{int1} as
\begin{equation}\label{abcd}
R(z) = \le(\begin{array}{cc}
a(z) &b(z)\\
c(z) & -a(z)
\end{array}\ri),
\end{equation}
and define
\be
Q_0 (z) \equiv -\det R(z) = a(z)^2+b(z)c(z).
\la{Q0}
\ee
We will assume all zeros $x_j,\,j=1,\dots,2g+2$ of $Q_0$ to be simple; then $Q_0$ takes the form
\be
Q_0(z) = C_0\frac{\prod_{j=1}^{2g+2}(z-x_j)}{\prod_{j=1}^{g+2} (z-z_j)^2},
\ee
where $C_0$ is a proportionality constant. Consider the hyperelliptic curve $\CC$ of genus $g$ with branch points at $x_1,\dots x_{2g+2}$ defined by
\be
\mu^2=Q_0(z)\;,
\la{spcurve}
\ee
and introduce the following meromorphic differential of the third kind, with $2g$ poles  on $\CC$:
\be
v=\mu(z) d z.
\la{defv}
\ee
The horizontal trajectories of $v$ generically connect its zeros $x_j$ with its poles $z_j$; denote the resulting critical graph by $\Gamma$. From the graph $\Gamma$ one can construct two graphs embedded in
the Riemann sphere: the graph $\Sigma$ with vertices at $z_j$ whose faces are triangles, and the tri-valent graph $\Sigma^*$ dual to $\Sigma$ with tri-valent vertices at $x_j$, as in Figure \ref{FigTrian}.

\begin{figure}[htb]
\begin{center}
\includegraphics[width=0.5\textwidth]{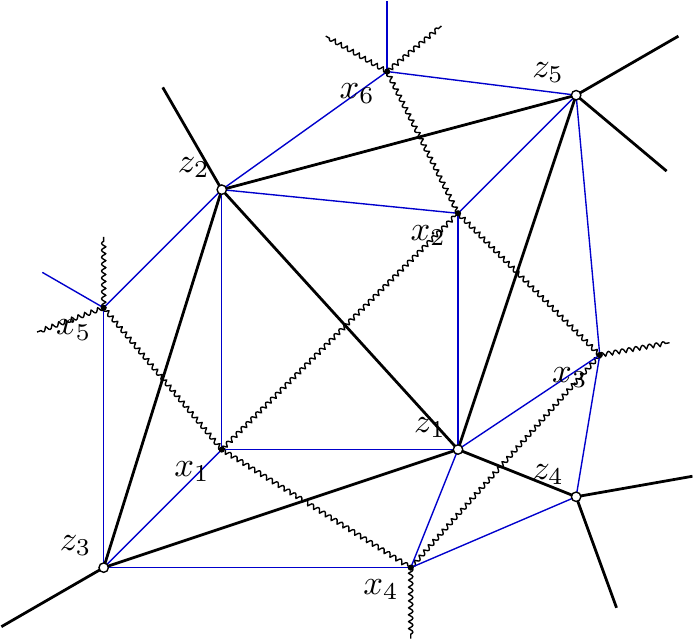}
\end{center}
\caption{Horizontal critical trajectories (blue lines) connect poles $z_j$ with zeros $x_k$ of $Q$ and form the
critical graph $\Gamma$. Black edges connecting poles $z_j$ form the graph $\Sigma$ (the triangulation of $\CC$) while the  zigzag curves connect zeros $x_j$ and form the dual tri-valent graph $\Sigma^*$.
All edges of  $\Sigma^*$ can be chosen to be  the branch cuts of $\CC$.  }
\label{FigTrian}
\end{figure}

The 
(logarithmic) complex shear coordinates on $V$ are assigned to edges of $\Sigma^*$; according to
\cite{BK_TMF} these coordinates can be extended by linearity to get homological shear coordinates 
assigned to elements of $H_1(\CC)$.  Introduce a canonical basis of cycles $(a_j,b_j)_{j=1}^g$ on $\CC$ with the intersection index 
and consider the corresponding set of logarithmic homological shear coordinates  $(\rho_{a_j},\rho_{b_j})$
on $V^{\textbf{r}}$ (see App of \cite{BK_TMF}), among which there are $g$ relations. In terms of $(\rho_{a_j},\rho_{b_j})$, Goldman's symplectic form $\Omega_G$ on $V^{{\bf r}}$ looks as follows:
\be
\Omega_G ={}\sum_{j=1}^g \d \rho_{a_j}\wedge \d\rho_{b_j}
\la{Gold}
\ee
The corresponding symplectic potential on $V^{{\bf r}}$ satisfying $\d \theta_G=\Omega_G$ will be chosen as follows:
\be\label{eq:thetaG}
\theta_G =\f{1}{2}\sum_{j=1}^g  (\rho_{a_j} \d\rho_{b_j} - \rho_{b_j}\d  \rho_{a_j}).
\ee 
Denote the monodromy map by $\Fcal$. The  theorem  of \cite{Hitchin,KorSam,AlekMal2} states that
\be
2\pi i \Omega_{KK}= \Fcal^* \Omega_G,
\la{KKG}
\ee
so that there exists the generating function $\Gcal$ such that
\be
\d \Gcal=\Fcal^* \theta_G-2\pi i\theta_{KK}.
\la{defG}
\ee

In this paper we compute explicitly the first three non-trivial terms in the $\hbar$-expansion of $\Gcal$ for fixed times:
\begin{equation}
\Gcal=\f{\Gcal_{-2}}{\hbar^2}+ \f{\Gcal_{-1}}{\hbar}+ \Gcal_0+\dots\;
\la{Ghbar}
\end{equation}
in Theorems \ref{thm:Gm2}, \ref{thm:Gm1} and \ref{thm:G0}. Note that, due to the almost trivial $\hbar$ dependence of the linear system \eqref{int1}, the Kirillov-Konstant symplectic potential in \eqref{defG} contributes only at order $\mathcal{O}(\hbar^{-1})$, while the potental
$\theta_G$ contains terms of every order starting from $\hbar^{-2}$. The existence of the monodromy symplectomorphism $\mathcal{F}$ implies the highly nontrivial consequence that the WKB expansion of Goldman's symplectic potential $\theta_G$ gives a closed form at all orders except for $\mathcal{O}(\hbar^{-1} )$. The computation of the coefficients in the expansion \eqref{Ghbar} is based on the $\hbar$-expansion of the homological shear coordinates $\rho_\ell$ for 
$\ell\in H_1(\CC,\Z)$ \cite{AlBrid,BK_TMF}:
\be
\rho_{\ell} = \int_{\ell}\left(\f{v}{\hbar}+ v_0+ \hbar v_1+\dots\right)
\la{rhohbar}
\ee
where $v_0, v_1, v_2,\dots$ are meromorphic differentials on $\CC$ arising from the WKB expansion. 
The integrals in (\ref{rhohbar}) are called Voros symbols after \cite{Voros}.

Our main result is the integration of definition \ref{eq:Gdef}, resulting in the explicit determination of $\Gcal_{-2},\,\Gcal_{-1},\,\Gcal_{0}$ in terms of contour integrals of WKB differentials $v,v_0,v_1$ on the WKB curve \eqref{spcurve}. The leading contribution to $\Gcal$ is

\be
\Gcal_{-2}=-\pi i \left(r_\infty {\rm reg}\int_{\infty^{(2)}}^{\infty^{(1)}} v +\sum_{j=1}^g r_j {\rm reg}\int_{z_j^{(2)}}^{z_j^{(1)}} v\right)
\la{regG2intro}
\ee
where $p^{(j)}$ denotes the image of the point $p\in\mathbb{P}^1$ on the $j$-th sheet of $\Ccal$, and the regularized integrals are defined by
\begin{equation}
{\rm reg}\int_{z_j^{(2)}}^{z_j^{(1)}} v:=\lim_{\epsilon\to 0}\left(\int_{z_j^{(2)}+\epsilon}^{z_j^{(1)} +\epsilon} v -2r_j\log\epsilon\right)
\end{equation}
and
\begin{equation}
{\rm reg}\int_{\infty^{(2)}}^{\infty^{(1)}} v:=\lim_{R\to \infty}\left(\int_{R^{(2)}}^{R^{(1)}}  v +2r_\infty \log R\right).
\end{equation}

The subleading term $\Gcal_{-1}$ is  given by the formula
\be
\Gcal_{-1}= \f{1}{2}\langle v, v_0\rangle   -2\pi i \sum_{j=1}^g q_j  A_j  -\pi i \sum_{k=1}^g B_k - \pi i  \sum_{j=1}^g j A_j,
\la{Gm1intro}\ee
where $A_j$, $B_j$ are the A- and B-periods of $v$, and $\langle\,,\,\rangle$ is the antisymmetric pairing defined by Riemann's bilinear relations
\be
\langle w, \tilde{w}\rangle:= \oint_{\partial \tilde{\Ccal}} \left(\int^x w\right) \tilde{w} 
\la{RBL}
\ee
where $\tilde{\Ccal}$ is the fundamental polygon of $\Ccal$. Finally, the constant term $\Gcal_0$ in the WKB expansion of the generating function is 

\be
 \Gcal_{0}=  -12\pi i \ln \tau_{B}(CP^1,Q_0) + F - \frac 1 2 \langle v_1,v\rangle.
 \ee
 \be
\frac{F}{i \pi } =  \frac 1 2 \sum_{j=1}^g \slint_{\l_j^{(2)} }^{\l_j^{(1)} }\!\!\!\! v_0  + \frac 1 2\, \slint_{\infty^{(2)} }^{\infty^{(1)} } \!\!\!\! \!\!v_0
+ \ln \frac {\prod_{a,k} (\l_a-z_k) } {\prod_j \mu_j\prod_{a<b} (\l_a- \l_b)}
\\
-\frac 1 {4r_\infty} \, \slint_{\infty^{(2)} }^{\infty^{(1)} }\!\!\!\!  v 
-
\sum_{k=1}^{g+2} \frac 1 {4r_k}\,  \slint_{z_k^{(2)} }^{z_k^{(1)} }  v  
.
\ee
Here $\lambda_j$'s are the location of the zeros of the $(2,1)$ entry of the matrix $R(z)$ \eqref{int1} , $\tau_B$ is Bergman's tau function (see Appendix \ref{secBerg})  and the regularization in the integrals is defined by
\be
\slint^{\l_j^{(1)} }_{\l_j^{(2)} }  v_0 := \lim_{p\to \l_j^{(1)} \atop q\to \l_j^{(2)} } \int_q^p v_0 - \frac 1 2 \ln (z(p)-\l_j)- \frac 1 2 \ln  (z(q)-\l_j).
\ee 
where $z:\mathcal C\to \C$ is the projection on the $z$--coordinate.\\
\noindent {\bf Acknowledgements.}
 The work of  MB was
supported in part by the Natural Sciences and Engineering Research Council of Canada (NSERC)  grant  RGPIN-2016-06660.  The work of  DK was
supported in part by the Natural Sciences and Engineering Research Council of Canada (NSERC)  grant  RGPIN-2020-06816.

\section{Second order equation and its WKB expansion}

The starting point of our discussion is the linear system \eqref{int1}, with coefficient matrix \eqref{abcd}. 
 Let us denote the zeros of the component $c(z)=R_{21}(z)$ by $\l_1,\dots,\l_g$.
Then, since  $R_\infty$ is diagonal, we have $c(z)\sim C/z^2$ as $z\to\infty$ and we can write:
\be
c(z) = C\frac {\prod_{j=1}^{g} (z-\l_j)}{\prod_{k=1}^{g+2} (z-z_k)}
\label{cdef}
\ee
for some constant $C$. If we also denote by
\be
\mu_j=a(\lambda_j)\;, \hskip0.7cm j=1,\dots, g\;
\la{muj}
\ee
we can write $a(z)$ as follows: 
\begin{equation}
a(z)=\sum_{j=1}^g\mu_j\prod_{k=1}^n\frac{\lambda_j-z_k}{z-z_k}\prod_{l\ne k}^g\frac{z-\lambda_l}{\lambda_j-\lambda_l}\; .
\end{equation}	

\begin{proposition}
Let $\psi(z;\hbar)$  be the second component of the vector-valued solution of \eqref{int1}; then the function $f= \sqrt{\frac{\hbar}{c(z)}}\psi(z;\hbar)$ satisfies the following second order ODE:
\begin{equation}\label{eq:Schro}
\f{d^2 f}{d z^2}-Q(z;\hbar) f=0
\end{equation}
where the {\rm potential} $Q(z;\hbar)$ is given by
\be
Q(z;\hbar) = \f{Q_0 }{\hbar^2}+ \f{ Q_1}{\hbar} + Q_2\;.
\la{Qz}
\ee
Here,
\be
Q_0 = -\det R = a^2+bc\;,
\la{defQ0}\ee
\be
Q_1 = a\f{{\rm d}}{{\rm d} z}\ln\le(\frac{c}{a}\ri)\;,
\la{defQ1}
\ee
\be\la{Q2}
Q_2 
= 
\frac 1 4 \le(\frac {c'}c\ri)^2 - \frac 1 2\le(\frac {c'}{c}\ri)'=-\f{1}{2} {\mathcal S}\left(\int^z c(x) \;{\rm d} x,\;z\right),
\ee
where ${\mathcal S}(f,z)$ is the Schwarzian derivative
\begin{equation}
\mathcal{S}(f,z)\equiv \left(\frac{f''}{f'}\right)'-\frac{1}{2}\left(\frac{f''}{f'} \right)^2.
\end{equation} 
\end{proposition}
\begin{proof}
Let us apply to the linear system \eqref{abcd} the following (singular) gauge transformation
\begin{align}
\la{gauge}
\Psi(z)=c^{-\frac{\sigma_3}{2}}\left( \begin{array}{cc}
1 & \frac{\hbar}{2}\frac{c'}{c}+a \\
0 & 1
\end{array} \right) \hbar^{\frac{\sigma_3}2}  F(z)\equiv g(z)F(z).
\end{align}
which maps the matrix $R(z)$ to the matrix
\begin{equation}
g^{-1}Rg-\hbar g^{-1} \f{{\rm d} g}{{\rm d} z} = \left( \begin{array}{cc}
0 &  Q \\
1 & 0
\end{array} \right),
\end{equation}
where
\begin{equation}
Q=\f{ a^2+bc}{\hbar^2} +\f{a\log\left(\frac{c}{a}\right)'}{\hbar}+\left[\frac{3}{4}\left(\frac{c'}{c}\right)^2-\frac{1}{2}\frac{c''}{c} \right].
\end{equation}
Then, the elements $F_{11},F_{12}$ of the matrix $F$ are the two independent solutions of \eqref{eq:Schro}.
\end{proof}

\subsection{Properties of $Q_0$, $Q_1$ and $Q_2$}

Let us discuss the properties of the meromorphic functions $Q_0$, $Q_1$ and $Q_2$.
\begin{itemize}
\item The function $Q_0$
 can be written as follows:
\begin{equation}\label{eq:Q0exp}
Q_0(z) = C_0\frac {P(z)}{\prod_{j=1}^{g+2}(z-z_j)^2}=\sum_{j=1}^{g+2} \left( \f{r_j^2}{(z-z_j)^2}+\f{H_j}{z-z_j}\right)
\end{equation}
where $P(z)=\prod_{j=1}^{2g+2}(x-x_j)$ is a polynomial of degree $2g+2$ and
\be
C_0=\sum_{j=1}^{g+2} r_j^2
\la{C0}
\ee

\item Function $Q_1$:

Notice that $a\sim - r_\infty /z+\dots$ as $z\to \infty$  while $c$ behaves as $C/z^2$. Therefore, $Q_1$ behaves as $ r_\infty z^{-2}$ as $z\to\infty$. If we write
\be
Q_1(z) = \sum_{j=1}^{g} \frac {\mu_j}{z-\l_j} 
+ \sum_{j=1}^{g+2} \frac {\g_j}{z-z_j}
\la{Q1g}
\ee
for some parameters $\gamma_j\in \C$, we get the following condition on the parameters entering in $Q_1$:
\be
\sum_{j=1}^{g} \mu_j + \sum_{j=1}^{g+2} \g_j=0\;.
\ee

\item Function $Q_2$:

The function $Q_2$ in \eqref{Q2} can be written as follows:
\be
Q_2 = \frac 1 4  \le(\sum_{j=1}^{g}\frac  1{(z-\l_j)}  - \sum_{k=1}^{g+2} \frac 1{z-z_k}\ri)^2
+\frac 1 2 \le(\sum_{j=1}^{g}\frac  1{(z-\l_j)^2}  - \sum_{k=1}^{g+2} \frac 1{(z-z_k)^2}\ri)\;,
\ee
and its Laurent expansion near $\l_\ell$ looks as follows:
\be\label{eq:Q2exp}
Q_2(z)= \frac 3{4(z-\l_\ell )^2} + 
 \frac {E_\ell}{z-\l_\ell }
   +F_\ell+\dots,
\ee
with 
\be\label{Fell}
E_\ell=
\frac 1 2 \le(
  \sum_{i\atop i\neq \ell } \frac 1{\l_\ell -\l_i}
   - \sum_{k=1}^n \frac1 {\l_\ell -z_k}\ri)\;,\hskip0.7cm
F_\ell
= \frac 1 4 \le(
\sum_{i\atop i\neq \ell} \frac 1 {\l_\ell-\l_i} -\sum_{k=1}^{g+2} \frac 1 {\l_\ell-z_k}
\ri)^2\;.
\ee
Near $z=\infty$ we have
\begin{equation}
	Q_2(z)=\frac{(\sum_j\lambda_j-\sum_kz_k)^2+2\sum_j\lambda_j^2-2\sum_kz_k^2}{4z^4}+O(z^{-5}).
\end{equation}

\end{itemize}


The potential (\ref{Qz}) of the resulting Schr\"odinger equation has second order poles at the points $z_1,\dots,z_{g+2}$ with biresidues $\hbar^{-2}r_1^2-\frac{1}{4} ,\dots, \hbar^{-2} r_{g+2}^2-\frac{1}{4}$ :
\be
Q=\left(\f{r_j^2}{\hbar^2}-\f{1}{4}\right)\f{1}{(z-z_j)^2}+ O((z-z_j)^{-1})\; , \hskip0.8cm z\to z_j ,
\la{Qzj}
\ee
it has the following behavior at $z\sim\infty$
$$
Q=\left(\f{r_\infty^2}{\hbar^2}+\f{r_\infty}{\hbar}\right)\f{1}{z^2}+\dots,
$$
and second order poles at the points $\l_1,\dots,\l_{n-3}$ with biresidues $3/4$:
\be
Q= \f{3/4}{(z-\l_j)^2}+\left(\frac{\mu_j}{\hbar}+E_j\right)\frac{1}{z-\lambda_j}+\mathcal{O}(1)
\la{Qlj}
\ee
The singularities at $\l_j$ are {\it apparent}\footnote{The terminology here is the one accepted in the specific  literature but it is a misnomer. According to the classical use of the term, an "apparent" singularity in an ODE is a point of singularity of the coefficients such that all the solutions are {\it analytic} in a neighbourhood thereof. Here, on the other hand, both solutions have a branchpoint with exponents $\pm \frac 1 2$. In general this could be a resonance, but the fact that there are no logarithms in the solution is the property that is termed improperly ``apparent". } i.e. the monodromy of the fundamental matrix of equation (\ref{eq:Schro}) around $\l_j$ is $-\1$. This can be  seen by inspection of the gauge transformation \eqref{gauge} and is a consequence of the following (for a proof, see \cite{BM}):
\begin{lemma}
\label{lemmaapp}
The ODE
\be\label{eq:LemmAppar}
f''(z) = \le(\frac 3{4 z^2} + \frac A z + B + \mathcal O(z)\ri)f(z)
\ee
has two linearly independent solutions with Frobenius exponents $\pm \frac 1 2$  at $z=0$ if and only if $A^2= B$.
\end{lemma}
The triviality (in $\mathbb P SL_2$) of monodromies at the points $\lambda_j$ translates to the following Bethe equations for the quantities $\gamma_k$ appearing in $Q_1$:

\begin{lemma}
The monodromy of the ODE \eqref{eq:Schro} around the apparent singularities $\l_j$ is $-\1$  if and only if  the following set of equations  are satisfied for 
coefficients $\gamma_j$, $j=1,\dots,g+2$  from (\ref{Q1g}):
\be
 \sum_{k=1}^{g+2} \frac {\gamma_k} {\l_\ell-z_k}
=\sum_{i\atop i\neq \ell} \frac {\mu_\ell- \mu_i}{\l_\ell-\l_i}
-\mu_\ell \sum_{k} \frac 1{\l_\ell-z_k} 
\label{Betther}
\ee
for $\ell=1,\dots,g$.
\end{lemma}
\begin{proof}
Use the behaviour \eqref{Qlj} of $Q$ at $\lambda_j$, and apply Lemma \ref{lemmaapp} to the Schr\"odinger equation \eqref{eq:Schro}:
\begin{equation}\label{eq:BetheQ}
\left(\frac{\mu_j}{\hbar}+\res{\lambda_j}Q_2 \right)^2=\frac{Q_0(\lambda_j)}{\hbar^2}+\frac{1}{\hbar}Q_1^{reg}(\lambda_j)+Q_2^{reg}(\lambda_j),
\end{equation}
where $\res{\lambda_j}Q_2 $ denotes the coefficient of $\frac{1}{z-\lambda_j} $ in the expansion of $Q_2$ around $\lambda_j$, as in equation $\eqref{eq:LemmAppar}$, and $Q_i^{reg}$ the regular part of the expansion at $\lambda_j$.  Using the explicit form of $Q_0,Q_1,Q_2$ in equations \eqref{eq:Q0exp}, \eqref{Q1g}, \eqref{eq:Q2exp}, we find
\bea
\le(\frac {\mu_\ell}{\hbar}
 +\frac 1 2 \sum_{i\atop i\neq \ell} \frac 1 {\l_\ell-\l_i}
 - \frac 1 {2}
 \sum_{k=1}^n \frac 1{\l_\ell-z_k}
\ri)^2
= \frac{Q_0(\l_\ell)}{\hbar^2} + \frac 1 \hbar \sum_{k=1}^{n} \frac {\gamma_k} {\l_\ell-z_k}
+ F_\ell
+ \frac 1 \hbar \sum_{i\atop i\neq \ell}\frac {\mu_{i}}{\l_\ell-\l_i}
\eea
On account that $\mu_\ell^2= Q_0(\l_\ell)$ and equation \eqref{Fell} we are left with the  Bethe equations for the quantities $\gamma_k$.
\end{proof}

\subsection{Canonical cover and WKB differentials}

We now introduce the WKB approximation of equation \eqref{eq:Schro}:
\be
f_{zz}- \left(Q_2+ \f{ Q_1}{\hbar} +\f{Q_0 }{\hbar^2}\right)f=0.
\la{fQ}\ee
Introduce the cover $\CC$ given by 
\begin{equation}\label{curve}
v^2=Q_0(z)dz^2.
\end{equation}
The curve $\CC$
 has $2g+2$ branch points which we denote by $x_1,\dots, x_{2g+2}$.  We denote the projection of $\CC$ to the $z$-plane by $f$, and the hyperelliptic involution on $\CC$ by $\nu$.
The homology group of the curve $\CC$, punctured at $2g+4$ points $f^{-1}(z_j)$,
can be represented as  a direct sum of even and odd components under the  involution $\nu$:
\be
H_1\left(\mathcal{C}\setminus\{f^{-1}(z_j)  \}_{j=1}^{g+2} \right)=H_+\oplus H_-, \hskip0.7cm\dim H_+= g+1\;.\hskip0.7cm  \dim H_-=3g+1.
\ee
Denote the points projecting to the poles $z_j$ by $z_j^{(1,2)}$:
$$
f^{-1}(z_j)=\{z_j^{(1)},z_j^{(2)}\}\;.
$$
The enumeration is chosen such that 
$$
\res{z_j^{(1)}} v= r_j\;,\hskip0.7cm 
\res{z_j^{(2)}} v= - r_j\;.
$$

Let small positively-oriented loops around points $\{z_j^{(1)}\}_{j=1}^{g+1}$ be denoted by $\{t_j\}_{j=1}^{g+1}$. Then generators of $H_+$ can be chosen to be $t_j^{+} = t_j+\nu_*t_j$, $j=1,\dots,g+1$.
The generators of $H_-$ can be chosen to be
\be
\{a_j,b_j\}_{j=1}^g\;, \hskip0.5cm \{t_j^{-} \}_{j=1}^{g+1},
\ee
where 
$$t_j^{-} =\f{1}{2}(t_j-\nu_* t_j)\;.$$

Introduce the divisor of degree $g$ given by
 \be
 D=\l^{(1)}_1+\dots+\l^{(1)}_g
 \ee
 where $\l^{(1)}_j=(\lambda_j,\mu_j)$ with $\lambda_j$ being the zeros of $c(z)$ and $\mu_j=a(\lambda_j)$.  Then $\lambda_j^{(2)}$ is the point having the same projection on $z$-plane but lying on another sheet of $\CC$ i.e. $\lambda_j^{(2)}=(\lambda_j,-\mu_j)$.

To study the limit $\hbar\to 0$ of equation \eqref{fQ} we introduce the asymptotic series   $s=\sum_{k=-1}^\infty \hbar^k s_k$  and
write the asymptotic series for the solution $f$ in the form
\be
f=v^{-1/2} \exp\left\{\int_{x_0}^x (\e^{-1} s_{-1}+s_0+\e s_1+\dots)v\right\},
\la{formWKB}
\ee
where $s_k$ are meromorphic functions on $\CC$  and $x_0$ is a basepoint.
We introduce also the meromorphic differentials 
\be
v_k=\f{1}{2}(s_k+\nu^* s_k) v\;,
\la{defvk}
\ee
The differential $v_k$ satisfies
\be
\nu^* v_k= -v_k.
\la{Vrel}
\ee

As a corollary of (\ref{fQ}) the  function $s$ satisfies the Riccati equation which in coordinate-independent form can be written as follows:
\be
  {\rm d} s+ v s^2 =- q v +\f{ Q_1 v }{\hbar}+ \f{v}{\hbar^2}
  \la{Ric}
\ee
where
$q=-Q_2 -\f{\Scal_v}{2}$ and $\Scal_v=\Scal(\int^z v, \cdot)$. Equivalently,  since 
$Q_2=-\f{1}{2} {\mathcal S}\left(\int^z c(x) \;{\rm d} x,\;\cdot\right)$, we can represent the meromorphic function $q$ on $\CC$ in the coordinate-independent form: 
\be
q=\f{1}{2}\left({\mathcal S}(\int^z c(x) \;{\rm d} x,\cdot )  - \Scal(\int^z v, \cdot)  \right)
=\f{1}{2}{\mathcal S}\left(\int^z c(x) \;{\rm d} x,\,\int^z v  \right)
\ee

Equivalently,  (\ref{Ric}) can be written as
$$
{\rm d}\left(\sum_{k=-1}^\infty \e^k s_k\right)+v\left(\sum_{k=-1}^\infty \e^k s_k\right)^2=- q v +\f{ Q_1 v }{\hbar}+ \f{v}{\hbar^2}
$$
The coefficients of  $\e^{-2}$, $\e^{-1}$ and $\e^0$ give
\be
s_{-1}=\pm  1 \;, \hskip0.7cm s_0=\frac{Q_1}{2 Q_0}\;, \hskip0.7cm
 s_1= - \f{s_0^2}{2}-\f{q}{2} -\f{{\rm d} s_0}{2v}
\la{sss}\ee
so that $v_{-1}=\pm v$. We shall choose the $"+"$ sign.

The higher functions $s_k$, $k\geq 1$ can be found recursively from relations
\be
{\rm d} s_k+ v\sum_{j+l=k\atop j,l\geq -1}  s_j s_l =0, \ \ \ k\geq 1\;,
\la{recur}
\ee
which imply
$$
 s_{k+1}=-\f{1 }{2s_{-1}}\left(\f{{\rm d} s_k}{ v}+\sum_{j+l=k,\atop j,l\geq 0} s_j s_l\right)\;. \ \ \ k\geq 1\;.
$$

The first three differentials in the WKB expansion are obtained from (\ref{sss}) taking into account that the only term in 
(\ref{sss}) which is not skew-symmetric under $\nu^*$ is $-{\rm d} s_0/2v$:
\be
\label{eq:WKBdiffs1}
v^2=Q_0(z) ({\rm d}z)^2\; ,\hskip0.7cm 
 v_0=\frac{1}{2v}Q_1(z) ({\rm d} z)^2\; ,
\ee

\be
\label{eq:WKBdiffs2}
v_1  =-\frac{v_0^2}{2v}+\frac{1}{2v}Q_2(z)({\rm d} z)^2+\frac{1}{4v} {\mathcal S} \left(\int^zv,z \right)({\rm d} z)^2.
\ee

\subsection{Properties of WKB differentials}\label{sec:WKBDiffs}

Here are the properties of the WKB differentials \eqref{eq:WKBdiffs1},\eqref{eq:WKBdiffs2}:

\begin{itemize}
\item
The differential of the third kind on $\CC$ given by
\be
v=\mu d z
\la{vdef}
\ee
has first order poles at points $z_j^{(1,2)}$ with the residues $\pm r_j$, and at $\infty^{(1,2)}$ with residues $\pm r_\infty$. The latter statement follows from the fact that near $\infty^{(1,2)}$ we have
\begin{equation}\label{eq:vinfty}
v\sim a(z){\rm d} z=\mp r_\infty \frac{dz}{z},
\end{equation}
\begin{equation}
\res{\infty^{(1,2)}} \;v=\pm r_\infty\;.
\end{equation}
\item
The differential of the third kind 
\be
v_0=\frac{Q_1}{2v}(d z)^2
\ee
is  holomorphic at branch points of $\CC$ (at the branch points, $v$ has second order zero, while the differential ${\rm d} z$ has a first order zero). Let us now discuss the behavior of $v_0$ near $\lambda_j^{(1,2)}$.  Near $\lambda_j^{(1)}$ we have $v\sim a(\lambda_j)d z$ and $Q_1\sim\frac{\mu_j}{z-\lambda_j}$. Therefore, near $\lambda_j^{(1,2)}$ one has the asymptotics
$$
v_0\sim {\pm}\frac{{\rm d} z}{2(\lambda-\lambda_j)}\;.
$$
Moreover, near $\infty^{(1,2)}$ we have 
$Q_1\sim\f{r_\infty}{z}$; thus, using  \eqref{eq:vinfty}, $v_0$ has simple poles at $\infty^{(1,2)}$ with residues ${\pm} 1/2$, respectively.
\item
We can write the differential $v_1$ (which turns out to be of second kind) as
\be
v_1= -\frac {v_0^2}{2v} + \f{Q_2}{2v}({\rm d} z)^2+ \f{1}{4v} {\mathcal S} \le(\int v, z\ri) ({\rm d} z)^2,
\la{defv1}
\ee
or more more symmetrically as 
\be
v_1= -\frac {v_0^2}{2v}+ \f{ ({\rm d} z)^2  }{4v}\left( {\mathcal S} (\int v, z) - {\mathcal S} \le(c , z\ri) \right)\;.
\la{v11}
\ee
Using the composition property of the Schwarzian derivative, it can also be written as follows:
\be
v_1= -\frac {v_0^2}{2v}-\f{v}{4} {\mathcal S} (c,\xi)
\la{v111}
\ee
where $\xi=\int^z v$ is the flat coordinate defined by differential $v$.

To prove (\ref{v111}) we used the composition rule for $c(\xi(z))$:
\begin{equation}
{\mathcal S} (c,z) ({\rm d} z)^2= {\mathcal S} (c,\xi)({\rm d}\xi)^2 + {\mathcal S} (\xi,z)({\rm d} z)^2\;.
\end{equation}

The differential $v_1$ has second order poles at branch points and second order poles  at $\infty^{(1,2)}$ with residues
\begin{equation}
    \res{\infty^{(1,2)}}{v_1}=\mp\frac{1}{4r_\infty}\;.
\end{equation}

\begin{lemma}
The conditions (\ref{Betther}) of triviality of monodromies around $\lambda_j$ are equivalent to equations
\be
{\rm res}|_{\l_j^{(1,2)}} v_1 =0\;.
\la{Beth1}
\ee

\end{lemma}
\begin{proof}
We use the expression
\begin{equation}
v_1=\frac{1}{2v}\left[\frac{1}{2}\mathcal{S}\left(\int v,z \right)+Q_2({\rm d}z)^2-v_0^2 \right].
\end{equation}
The term involving the Schwarzian derivative can be ignored since it   is regular at $\lambda_j$.  Using the expansions at $z=\lambda_j$
\begin{equation}
Q_0(z)=\mu_j^2+(z-\lambda_j)Q_0'(\lambda_j)+\mathcal{O}(z-\lambda_j)^2,
\end{equation}
\begin{equation}
Q_1(z)^2=\frac{\mu_j^2}{(z-\lambda_j)^2}+\frac{2\mu_j}{z-\lambda_j}Q_1^{reg}(\lambda_j)+\mathcal{O}(1),
\end{equation}
\begin{equation}
Q_2(z)=\frac{3}{4}\frac{1}{(z-\lambda_j)^2}+\frac{E_j}{z-\lambda_j}+\mathcal{O}(1).
\end{equation}
and expressing $v$ and $v_0$ in terms of $Q_0,Q_1$ we have
\begin{equation}
\frac{1}{2v}\left[Q_2(dz)^2-v_0^2 \right]=\pm\frac{1}{4\mu_j}\left[\frac{1}{(z-\lambda_j)^2}+\frac{1}{(z-\lambda_j)}(2\mu_j E_j -Q_1^{reg}(\lambda_j))+\mathcal{O}(1) \right].
\end{equation}
The vanishing of the residue is equivalent to the $O(\hbar^{-1})$ equation of \eqref{eq:BetheQ}.
\end{proof}
\end{itemize}

Let us summarize in a table all the properties of the WKB differential discussed in this section:
\begin{center}
\begin{tabular}{ c c c c }
Differential    & Pole                & Order of the pole   & Residue \\
\hline  
 $v=\mu dz$     & $z_j^{(1,2)}$       & $1$                 & $\pm r_j$ \\
                & $\infty^{(1,2)}$    & $1$                 & $\pm r_\infty$\\
\hline
$v_0=\frac{Q_1}
{2v}(dz)^2 $    & $z_j^{(1,2)}$       & $1$                 & $\pm\frac{1}{2}$ \\
                & $\lambda_j^{(1,2)}$ & $1$                 & $\pm\frac{1}{2}$ \\
\hline
$v_1=-\frac{v_0^2}{2v}-
\frac{v}{4}\mathcal{S}
(c,\xi)$        & $z_j^{(1,2)}$       & $2$                 & $\mp\frac{1}{4r_\infty} $\\
                & $\infty^{(1,2)}$    & $2$                 & $\mp\frac{1}{4r_k}$ \\
                & $\lambda_j^{(1,2)}$ & $2$                 & $0$ \\
                & $x_j$               & $2$                 & 
\end{tabular}
\end{center}

\subsection{WKB expansion of complex shear coordinates}

Let us assume that $Q_0$ is a Gaiotto-Moore-Nietzke (GMN) differential (i.e. it has no horizontal trajectories connecting two zeros \cite{GMN}). Then  the cover  \eqref{curve} (also called the "WKB curve") admits a canonical triangulation $\Sigma$, constructed as follows. Due to the genericity assumption on $Q_0$, the horizontal trajectories always start at a zero $x_j$ and end at a pole $z_k$: let us call by $\Gamma$ the critical graph whose edges are such trajectories. The WKB triangulation $\Sigma$ is defined as the graph having the poles $z_k$'s as vertices, with every face containing exactly one zero of $Q$. The dual graph to $\Sigma$, whose vertices are instead the zeros $x_j$, will be denoted by $\Sigma^*$. This construction is shown in Figure \ref{FigTrian}.

To each edge $e$ of the graph $\Sigma$ one assigns a coordinate $\rho_e\in\mathbb{C}$ whose exponential is a complex shear coordinate, the simplest example of a Fock-Goncharov coordinate (see Appendix A.2 of \cite{BK_TMF} for more details).

 To every edge $e$ of the triangulation it is possible to associate a cycle $\ell_e\in H_-$, defined to be the loop that goes clockwise around the edge $e^*$ of $\Sigma^*$; to the cycle $e$ one assigns the Fock-Goncharov coordinate $\rho_e$
$$
\{\rho_e,\rho_{e'}\}=e\circ e'
$$

Because the Poisson bracket is constant, it is possible to choose linear combinations of the $\zeta_e$'s that are Darboux conjugate. These  
 coordinates extend by linearity  to the $a_j$ and $b_j$ cycles, that we denote by $\rho_{a_j},\rho_{b_j}$,
 and their Goldman bracket is
\begin{equation}
\left\{\rho_{a_j},\rho_{b_k} \right\}_G=\frac{\delta_{jk}}{2}\; .
\end{equation}

The following proposition is an analog of  Prop. 5.2  in \cite{BK_TMF}. 
It shows  that the homological shear coordinates $\rho_\ell$ (for $\ell\in H^-(\mathcal C)$) admit an  asymptotic expansion in terms of periods of $S_{odd}$.
\begin{proposition}
The formal asymptotic expansion of the  homological shear coordinate $\rho_{\ell}$
for each $\ell\in H_-$ looks as follows:
\be
\rho_{\ell}\sim\f{1}{\hbar} \int_\ell v+ \int_\ell v_0 +\hbar\int_\ell v_1+\dots\ \ .
\la{rhoexp}
\ee
where $v_k$ are given by \eqref{defvk}.
%
The relation (\ref{rhoexp}) is understood in $PSL(2)$ sense i.e. up to an addition of   $
\pi i k$ for $k\in\Z$.%
\end{proposition}
{\it Sketch of the proof.} The proof is parallel to the proof of Prop. 5.2 of \cite{BK_TMF}. The difference is the presence of the term $Q_1$ in this paper which was absent in \cite{BK_TMF}; moreover, $Q_1$ has additional singularities at $\lambda_j$
which are apparent singularities of the equation (\ref{eq:Schro}) were absent in \cite{BK_TMF} and also in \cite{Alleg}. In the frameworks of  \cite{BK_TMF} and \cite{Alleg} the differential $V$ contains only odd powers of $\hbar$ (since in these papers it is assumed that $Q_1=0$) while in our present context we also have all even powers.


The presence of apparent singularities at $\lambda_j$ does not modify the asymptotics (\ref{rhoexp}) by the following reason:
the issue of apparent singularities at the poles of $Q_1$   is completely analogous to the  case of the Lax pair for the sixth Painlev\'e\ equation (i.e. our case with $g=1$); we refer to \cite{kawai2005algebraic}, Theorem 4.4. The reason why there are no different Stokes' regions in the WKB analysis near a $\l_j$ is that, up to a (Borel resummable $\hbar$--dependent) conformal change of coordinate $z\to \xi(z,\hbar)$, the local model of the equation in a neighbourhood of $z=\l_j$ is 
\be
-\frac {{\rm d}^2}{{\rm d}\xi^2} f +\le( \frac {4\xi^2}{\hbar^2} - \frac 3{4\xi^2} \ri)f =0\;.
\ee

This equation, while formally displaying a double turning point at $\xi=0$ (corresponding to $z=\l_j$), does not, in fact, exhibit any Stokes' phenomenon since its general solution is explicitly written as
\be
f(\xi) = \frac A{\sqrt \xi}{\rm e}^{\frac 1 \hbar \xi^2 } +\frac B{\sqrt \xi}{\rm e}^{-\frac 1 \hbar \xi^2 }.
\ee
This implies that the analysis of \cite{BK_TMF} goes through without further modifications. In the case of the opers relevant for all the Painlev\'e\ equations $I$--$VI$ this was discussed in detail in Chapter 4 of \cite{kawai2005algebraic} (and references therein). 

\QED

%


\section{ WKB expansion of the generating function }
\label{sec3}
The generating function $\Gcal$ defined by (\ref{defG}) has the following formal expansion in powers of $\hbar$:
\be
\Gcal=\f{\Gcal_{-2}}{\hbar^2}+ \f{\Gcal_{-1}}{\hbar}+ \Gcal_0+\dots\;.
\la{Ghbar1}
\ee 
In this section we compute the first three coefficients, $\Gcal_{-2}$, $\Gcal_{-1}$ and $\Gcal_0$. Let us introduce the following notation: for any two 1-forms $w$ and $\tilde{w}$ we consider  their periods
$(A_j,B_j)$ and $(\tilde A_j, \tilde B_j)$ and introduce the pairing 
\be\label{eq:RBPairing}
\langle w, \tilde{w} \rangle=\sum_{j=1}^g A_j \tilde{B}_j- B_j\tilde{A}_j
\ee

The expansion of  the symplectic potential $\theta_G$  from equation \eqref{eq:thetaG} can be written as follows using (\ref{rhoexp})  and the pairing \eqref{eq:RBPairing}:
\be
\theta_G=\f{\theta_{G}^{(-2)}}{\hbar^2}+ \f{\theta_G^{(-1)}}{\hbar}+ {\theta_{G}^{(0)}}+ O(\hbar)
\la{thetaGex}
\ee
where
\be
\theta_G^{(-2)}=\f{1}{2}\langle v,\delta v\rangle
\la{tG2}
\ee
\be
\theta_G^{(-1)}=\langle v_0,\delta v\rangle-\f{1}{2}\delta \langle v_0, v\rangle
\la{tG1}
\ee
\be
\theta_G^{(0)}=\f{1}{2}\langle v_0,\delta v_0\rangle+ \langle v_1, \delta v\rangle -\frac{1}{2}\delta \langle v_1, v\rangle.  
\la{tG0}
\ee

\subsection{Formula for $\mathcal{G}_{-2}$}

Chose a  set of generators of $H_1(\CC)$ which we denote by $\{a_j,b_j\}_{j=1}^g,\;\{t_j\}_{j=1}^{g-1}$.
Introduce the $a$ and $b$-periods of the differential $v$:  
$$A_j=\int_{a_j} v \;,\hskip0.7cm B_j=\int_{b_j} v$$
\begin{theorem}\label{thm:Gm2}
The equation \eqref{tG2} for $\mathcal{G}_{-2}$  can be written as
\be
\d \Gcal_{-2}= \f{1}{2}\sum_{j=1}^g  (A_k \d B_k -B_k\d  A_k)
\la{eqG2}
\ee
and its solution is given by
\be
\Gcal_{-2}=-\pi i \left(r_\infty {\rm reg}\int_{\infty^{(2)}}^{\infty^{(1)}} v +\sum_{j=1}^g r_j {\rm reg}\int_{z_j^{(2)}}^{z_j^{(1)}} v\right)
\la{regG2}
\ee
where
$$
{\rm reg}\int_{z_j^{(2)}}^{z_j^{(1)}} v=\lim_{\epsilon\to 0}\left(\int_{z_j^{(2)}+\epsilon}^{z_j^{(1)} +\epsilon} v -2r_j\log\epsilon\right)
$$
and
$$
{\rm reg}\int_{\infty^{(2)}}^{\infty^{(1)}} v=\lim_{R\to \infty}\left(\int_{R^{(2)}}^{R^{(1)}}  v +2r_\infty \log R\right)
$$
\end{theorem}
\begin{proof}In the coordinate system $(A_j,\lambda_j)$, the form in the r.h.s. of (\ref{eqG2}) has only $\d A_j$ - contributions since it depends only on the moduli of $\CC$, and not on the point of Jacobian. One can write $v$ as follows,  using the properties of $v$ that we listed in Section \ref{sec:WKBDiffs}:
\be
v=r_\infty w_{\infty^{(2)}, \infty^{(1)}}+ \sum_{j=1}^g r_j w_{z_j^{(2)}, z_j^{(1)}} + \sum_{j=1}^g A_j u_j
\la{vwv}
\ee
where $w_{x,y}$ is the differential of third kind on $\CC$ with residues $-1$ and $+1$ at $x$ and $y$, respectively, normalized by the condition of vanishing $a$-periods,  and $u_j$ is the holomorphic differential normalized via $\int_{a_k} u_j=\delta_{jk}$. The $b$-period of $w_{x,y}$ is given by
\begin{equation}
    \oint_{b_j}w_{x,y}=2\pi i\int_{x}^{y}u_j\;,
\end{equation}
so that

\begin{equation}\label{eq:BcycGm2}
B_k=\int_{b_k}v = 2\pi i\left(  r_\infty \int_{\infty^{(2)}}^{ \infty^{(1)}} u_k+  \sum_{j=1}^g r_j \int_{z_j^{(2)}}^{ z_j^{(1)}} u_k\right) +\sum_{j=1}^g A_j \Omega_{jk}
\end{equation}

Moreover, since $z_j$ and $r_j$ are independent of the periods $A_j$,  from \eqref{vwv} and \eqref{eq:BcycGm2} we have
\be
\f{\delta v}{\delta A_j}  = u_j\;, \hskip0.7cm  \f{\d B_k}{\d A_j} = \Omega_{jk}
\la{Ajv}
\ee
and, therefore,
\begin{equation}
\sum_{k=1}^g  (A_k \d B_k -B_k\d  A_k)= -2\pi i \sum_{k=1}^g  \left(r_\infty \int_{\infty^{(2)}}^{\infty^{(1)}} u_k+ \sum_{j=1}^g r_j \int_{z_j^{(2)}}^{ z_j^{(1)}} u_k\right) \d A_k 
\end{equation}
which, due to (\ref{Ajv}),  equals to 
$$
 -2\pi i\d  \left(r_\infty {\rm reg} \int_{\infty^{(2)}}^{\infty^{(1)}} v+\sum_{j=1}^g r_j {\rm reg} \int_{z_j^{(2)}}^{ z_j^{(1)}} v\right)
$$
leading to (\ref{regG2}).
\end{proof}

\subsection{Formula for $\mathcal{G}_{-1}$}

We shall use the  version of Riemann bilinear relations (see \cite{Dub}, eq. (2.5.6)) given in \eqref{RBL}.

The first theorem we need is the following: let $\mu_j=a(\lambda_j)$ and 
\be
\l_j^{(1)} = (\lambda_j,\mu_j)
\la{pj}
\ee
and consider the divisor $D=\l_1^{(1)}+\dots+\l_g^{(1)}$.

According to Th.\ref{TKK} the symplectic potential for Kirillov-Kostant symplectic form can be expressed as follows in terms  $a$-periods of $v$ and the divisor $D$ as follows:
\be
\sum_{j=1}^g \tr L_j G_j^{-1} \delta G_j = \sum_{j=1}^g A_j \delta q_j =- \left( \sum_{j=1}^g q_j \delta A_j\right)+ \delta\left( \sum_{j=1}^g q_j  A_j\right)
\la{tKK1}
\ee
where
\be
q=\Acal_{\infty^{(2)}}(D)-K^{\infty^{(2)}}\;,
\la{qD}
\ee
 where $K^{\infty^{(2)}}$ is the vector of Riemann constants at the point $\infty^{(2)}$.

Using this fact we shall prove the following formula for $\Gcal_{-1}$:

\begin{theorem}\label{thm:Gm1}
The equations for $\Gcal_{-1}$ look as follows
\be
\delta \Gcal_{-1}=  \langle v_0, \delta v\rangle- 2\pi i\sum_{j=1}^g \tr L_j G_j^{-1} \delta G_j  + \f{1}{2}\delta \langle v, v_0\rangle
\la{eqGm1}
\ee
and  the solution is  given by the formula
\be
\Gcal_{-1}= \f{1}{2}\langle v, v_0\rangle   -2\pi i \sum_{j=1}^g q_j  A_j  -\pi i \sum_{k=1}^g B_k - \pi i  \sum_{j=1}^g j A_j
\la{Gm1}\ee
\end{theorem}

{\it Proof.} 
 Equation \eqref{eqGm1} follows from $\delta\mathcal{G}=\theta_G-2\pi i\theta_{KK}$, together with equations \eqref{tG1} and the \eqref{thetaKK}. Let us work in the coordinate system $(A_j,\lambda_j)$. Then   $\delta_{\lambda_j} v=0$. 
On the other hand, since $\pa_{A_j} v= u_j$ we compute
$$
 \langle v_0, \pa_{A_j} v\rangle=  \langle v_0, u_j\rangle = - \oint_{\partial \tilde{\Ccal}} \left(\int^x u_j\right) v_0
$$
\be
=-2\pi i \sum {\rm res} \left(\int^x u_j\right) v_0 =-\left(\pi i  \sum_{k=1}^g \int_{\l_k^{(2)}}^{\l_k^{(1)}} u_j +\pi i  \int_{\infty^{(2)}}^{\infty^{(1)}} u_j\right),
\la{v0vj}
\ee
where we used that $v_0$ has residues $\pm\frac{1}{2}$ at $\lambda_j^{{(1,2)}}$ and $\infty^{(1,2)}$, respectively. Therefore, 
\begin{equation}
\Gcal_{-1}= \f{1}{2}\langle v, v_0\rangle - 2\pi i\sum_{j=1}^g q_j  A_j  +f,
\end{equation}
where
$$
\delta f=  \langle v_0, \delta v\rangle + 2\pi i \sum_{j=1}^g q_j \delta A_j
$$
Therefore, $\delta_{\lambda_j} f=0$ and 
$$
\pa_{A_j} f=\langle v_0, u_j \rangle  +2\pi i q_j
$$
Using (\ref{v0vj}) and (\ref{qD}) we get 
\begin{equation}
\pa_{A_j} f=-\pi i\left(\sum_{k=1}^g \int_{\l_k^{(2)} }^{\l_k^{(1)} } u_j + \int_{\infty^{(2)} }^{\infty^{(1)} } u_j\right)+2\pi i \left(\sum_{k=1}^g\int_{\infty^{(2)} }^{\l^{(1)} _k} u_j -K^{\infty^{(2)} }_j\right)
\end{equation}
Choose a branch point $x_1$ as the corner of the fundamental polygon. Then this equation can be written as
\be
\f{1}{2\pi i}\pa_{A_j} f=-\left(\sum_{k=1}^g \int_{x_1}^{\l_k^{(1)}} u_j +\int_{x_1}^{\infty^{(1)} } u_j\right) + \sum_{k=1}^g \int_{x_1}^{\l_k^{(1)}} u_j + g\int_{\infty^{(2)} }^{x_1} u_j -K^{\infty^{(2)} }_j
\ee
or
\be
\f{1}{2\pi i} \pa_{A_j} f=-(g-1) \int_{\infty^{(2)} }^{x_1} u_j-K^{\infty^{(2)} }_j
\ee
Therefore, we get
$$
\f{\delta f}{\delta A_j} =  -2\pi i K^{x_1}_j 
$$
Using the representation
$$
K^{x_1}_j =\f{j}{2}+\f{1}{2}\sum_{k}\Omega_{jk}
$$
and relation
$$
\Omega_{jk}=\f{\delta B_k}{\delta A_j}
$$
we have
\be
f=-\pi i \sum_{j=1}^g j A_j -\pi i \sum_{k=1}^g B_k
\la{fexp}\ee

\QED

\subsection{Formula for $\Gcal_0$}

\begin{theorem}\label{thm:G0}
The equations for $\Gcal_{0}$ looks as follows
\be
\delta \Gcal_{0}=\theta_G^{(0)} =\f{1}{2}\langle v_0,\delta v_0\rangle+ \langle v_1, \delta v\rangle -\frac 1 2 \delta \langle v_1, v\rangle  
\la{eqGm0}
\ee
and  the solution is  given by the formula 
\be
 \Gcal_{0}=  -12\pi i \ln \tau_{B}(CP^1,Q_0) + F - \frac 1 2 \langle v_1,v\rangle
 \ee
 where
 \be
\frac{F}{i \pi } =  \frac 1 2 \sum_{j=1}^g \slint_{\l_j^{(2)} }^{\l_j^{(1)} }\!\!\!\! v_0  + \frac 1 2\, \slint_{\infty^{(2)} }^{\infty^{(1)} } \!\!\!\! \!\!v_0
+ \ln \frac {\prod_{a,k} (\l_a-z_k) } {\prod_j \mu_j\prod_{a<b} (\l_a- \l_b)}
\\
-\frac 1 {4r_\infty} \, \slint_{\infty^{(2)} }^{\infty^{(1)} }\!\!\!\!  v 
-
\sum_{k=1}^{g+2} \frac 1 {4r_k}\,  \slint_{z_k^{(2)} }^{z_k^{(1)} }  v  
,
\ee
and $\tau_B(CP^1,Q_0)$ is the Bergman tau-function defined by (\ref{Bergtau}).
The regularization is the ``$z$''--regularization, where we subtract the singular part computed in the $z$--coordinate. Specifically,
\be
\slint^{\l_j^{(1)} }_{\l_j^{(2)} }  v_0 = \lim_{p\to \l_j^{(1)} \atop q\to \l_j^{(2)} } \int_q^p v_0 - \frac 1 2 \ln (z(p)-\l_j)- \frac 1 2 \ln  (z(q)-\l_j)
\ee
\end{theorem}
\noindent {\it Proof.}
Let us write the potential $\theta_G^{(0)}$ as follows:
\be
\theta_G^{(0)}=  \Theta^{(0)} +\Bil {\widehat 
{v}_1}{\delta v}+\delta \langle v, v_1\rangle 
\la{th0G}
\ee
where
\be
\label{potential}
\Theta^{(0)}  = \frac 1 2 \Bil{v_0}{\delta v_0} +\Bil {\wt {v}_1}{\delta v} \;.
\ee
Here we have set
\be
\label{vvt}
\widehat{v}_1=\f{ (d z)^2  }{2v}{\mathcal S}\le(\int v, z\ri),\ \ \ \ 
\widetilde{v}_1=-\frac {v_0^2}{2v}-\f{ (d z)^2  }{2v}{\mathcal S}(c, z)
\ee
so that the differential  $v_1$ is the sum
\be
v_1=\wt {v}_1+ \widehat{v}_1
\ee
The integration of the last term in (\ref{th0G}) is trivial. The second term can be integrated in terms of the  Bergman tau-function described in the Proposition \ref{prop43}
The integration of $\Theta^{(0)}$ is provided in Lemma \ref{lemma41}.
\QED
\begin{proposition}[Proposition \ref{tauprop} ]
\label{prop43}
The solution to the equation
\be
-12\pi i \, \delta\ln \tau_B(\CC,{\rm d} z)=\langle \widehat{v}_1, \delta v\rangle
\ee
is given by  the {\rm Bergman tau-function}:
\be
\tau_B=
\exp\left\{-\f{1}{6} \sum_{k=1}^{g+2} \f{1}{r_k} {\rm reg} \int_{x_1}^{z_k} v  \right\}
 \prod_{j<k}(x_j-x_k)^{5/144}  \prod_{j<k}(z_j-z_k)^{1/6} \prod_{j,k}(x_j-z_k)^{-7/72}
\la{Bergtau}
\ee
where the regularized integrals are defined by (\ref{regint}).
\end{proposition}

The integration of the term $\Theta^{(0)}$ is provided by the following lemma:

\begin{lemma}
\la{lemma41}
The solution to equation
\be
\delta F=\Theta^{(0)}
\ee
is given by  

\be
\frac{F}{i \pi } =  \frac 1 2 \sum_{j=1}^g \slint_{\l_j^{(2)} }^{\l_j^{(1)} }\!\!\!\! v_0  + \frac 1 2 \slint_{\infty^{(2)} }^{\infty^{(1)} } \!\!\!\! \!\!v_0
+ \ln \frac {\prod_{a,k} (\l_a-z_k) } {\prod_j \mu_j\prod_{a<b} (\l_a- \l_b)}
\\
-\frac 1 {4r_\infty}  \slint_{\infty^{(2)} }^{\infty^{(1)} }\!\!\!\!  v
-
\sum_{k=1}^{g+2} \frac 1 {4r_k}\,  \slint_{z_k^{(2)} }^{z_k^{(1)} }  v  \;
.
\ee
\end{lemma}
{\it Proof.}
The computation can be performed in any set of coordinates; we choose to use coordinates $\lambda_j's$ and periods $\{A_j\}$ (alternatively one could use any any moduli of the spectral curve). Then the spectral curve is $\lambda_j$--independent. 
By application of the Riemann bilinear identities we obtain
\be
\frac 1{i\pi} \Theta^{(0)} = \sum \res{} \le(   v_0(x) \int^x \delta v_0 +  2\widetilde{v}_1 \int^x \delta v\ri).
\label{resid}
\ee
where the sum extends over all poles of the expression in the residue bracket.  These  poles are located  at:
the points $z_j^{(1,2)}$,
the branch points $x_j$ of the curve $\CC$,
 the  points $\lambda_j^{(1,2)}$,
 the points $\infty^{(1,2)}$.

\paragraph{Branch points.} We are going to show that the residues at the branchpoints of the two terms in \eqref{resid} cancel each other. To this end we observe that from 
\be
v_0 = \frac {Q_1 {\rm d} z^2}{2v}, \label{v0}
\ee
we have 
\be
\delta v_0 = \frac {\delta Q_1 ({\rm d} z)^2}{2v}  - \frac {Q_1({\rm d} z)^2}{2v^2} \delta v= \frac {\delta Q_1 ({\rm d} z)^2}{2v}  -\frac{v_0}{v} \delta v
\label{dv0}
\ee
The first term does not have poles at the branch points because ${\rm d} z^2$ and $v$ both have double zeros. The second term is present only for differential in the moduli of the curve (leaving $\l_j$'s constant). 
Secondly, in the expression for $\wt v_1$ the only term with poles at the branch point are
\be
\label{v1} 
\wt v_{1} = -\frac {v_0^2}{v} +\mathcal O(1).
\ee
Thus, 
\be
 \sum_{b.pts} \res{x_k} \le(
v_0(p) \int^p \pa_{A_j} v_0 + 2 \wt v_1 \int^p \pa_{A_j} v
\ri) \mathop{=}^{\eqref{dv0}}\sum_{b.pts} \res{x_k} \le(
-v_0(p) \int^p \frac{v_0}{v}\pa_{A_j} v + 2 \wt v_1 \int^p \pa_{A_j} v
\ri)
\ee

\be
= \sum_{b.pts} \res{x_k} \le(
-v_0(p) \int^p \frac{v_0}{v}\pa_{A_j} v - \frac{ v_0^2}{v} \int^p \pa_{A_j} v
\ri)
\label{resk}
\ee
The computation of this residue is easier if done in the local coordinate $\zeta $ given by  $z = x_k + \zeta^2$: in this coordinate each of the differentials (being all odd under the hyperelliptic involution) are expressed as functions of $\zeta^2$. We denote 
\be
v_0 = f_0(\zeta^2) {\rm d} \zeta,\ \ \ v = \zeta^2 h(\zeta^2){\rm d} \zeta, \ \ \ 
\pa_{A_j} v = g(\zeta^2) {\rm d} \zeta,
\ee
where we have used that $v$ has a double zero at $\zeta=0$ (and $h(0)\neq 0$).
We also observe that in the computation of the residues the base point of integration of the integrals  is irrelevant because it adds a constant and this yields  no residue.  Then we can represent the indefinite integral 
$\int^p \frac {v_0}{v} \pa_{A_j} v$ as a locally defined meromorphic  {\it odd} function with a simple pole at $\zeta=0$. In explicit terms we have 
\be
\eqref{resk} = -\res{\zeta=0}  \le( -  f_0(\zeta^2) \frac {f_0(0) g(0)}{\zeta h(0)} + \frac {f_0(0)^2}{\zeta^2 h(\zeta^2)} \int_0^\zeta g(\xi^2){\rm d}\xi + \mathcal O(1)\ri) {\rm d} \zeta = 0.
\ee
\paragraph{Contribution of the other residues.}
Consider first  one of the moduli,  $A_j$, of the curve that does not modify the $z$--projection of the divisor $D$.


Recall that $\res{\l_j^{(1,2)}} v_0= \res{\infty^{(1,2)} } v_0= \pm \frac 1 2$ and thus $\pa_{A_j} v_0$ is locally analytic at the points $\l_j$ and $\infty$; the same applies to $\pa_{A_j} v$. Viceversa, from \eqref{v1} it follows that $\wt  v_1$ has a double pole at $\l_j$ with coefficient 
\be
\wt v_1 = {\pm} \frac 1{4 \mu_j}  \frac 1{(z-\l_j)^2} + \mathcal O(1),  \  \  \text { near  } \l_j^{(1,2)} 
\ee
where we emphasize the absence of residue. 
Finally we need the residues of $\wt v_1$ at the points  $z_k^{(1,2)}  $'s and $\infty^{(1,2)} $; a short computation using \eqref{vvt} yields
\bea 
\res{z_k^{(1,2)}  } \wt v_1 = \mp  \frac 1 {8r_k} \ ;\ \ \ \ \res{\infty^{(1,2)}  } \wt v_1 = \mp \frac 1 {8r_\infty}\;.
\eea 
Keeping this in mind, the result is  then
\bea
 \sum_{\l_j^{(1,2)} , \infty^{(1,2)} , z_k^{(1,2)}  }& \res{} \le(
v_0(p) \int^p \pa_{A_j} v_0 + 2 \wt  v_1 \int^p \pa_{A_j} v
\ri)=
\frac 1 2\sum_j \int_{\l_j^{(2)} }^{\l_j^{(1)} } \pa_{A_j} v_0  + \frac 1 2 \int_{\infty_-}^{\infty+} \pa_{A_j} v_0
\\
& - \sum_j\frac 1 {2\mu_j} \frac{\pa_{A_j} v }{{\rm d} z}\bigg|_{\l_j^{(2)} }^{\l_j^{(1)} } 
-\frac 1 {4r_\infty}  \int_{\infty^{(2)} }^{\infty^{(1)} }\!\!\!\! \pa_{A_j} v
-
\sum_{k=1}^{n+2} \frac 1 {4r_k} \int_{z_k^{(2)} }^{z_k^{(1)} } \pa_{A_j} v
\\
=&
\frac 1 2\sum_j \int_{\l_j^{(2)} }^{\l_j^{(1)} } \pa_{A_j} v_0  + \frac 1 2 \int_{\infty_-}^{\infty+} \pa_{A_j} v_0\,
{-}\, {\pa_{A_j}  }\ln\prod_{k=1}^{g+2} \mu_k  
{-\frac 1 {4r_\infty} } \int_{\infty^{(2)} }^{\infty^{(1)} }\!\!\!\! \pa_{A_j} -
\sum_{k=1}^{g+2} \frac 1 {4r_k} \int_{z_k^{(2)} }^{z_k^{(1)} } \pa_{A_j} v
\eea
where we have used that $\frac{\pa_{A_j} v}{{\rm d} z}\bigg|_{\l_j^{(1,2)}  } = \pm \pa_{A_j}\mu_j $.
We now observe that the derivative is made at $z$--value fixed under the integral sign; therefore we can pull the derivative outside provided we interpret the integration as a $z$--regularized integral:
\bea
 \sum_{j}& \res{\l_j^{\pm}} \le(
v_0(p) \int^p \pa_{A_j} v_0 + 2 v_1 \int^p \pa_{A_j} v
\ri)
=\\
=&\pa_{A_j} \le(
\frac 1 2\sum_j \slint_{\l_j^{(2)} }^{\l_j^{(1)} }  v_0  
+
 \frac 1 2 \slint_{\infty_-}^{\infty+}  v_0 
 {-}
  \ln \prod_{k=1}^{g+2} \mu_k  
{-\frac 1 {4r_\infty} } \slint_{\infty^{(2)} }^{\infty^{(1)} }\!\!\!\!  v
-
\sum_{k=1}^{n+2} \frac 1 {4r_k} \slint_{z_k^{(2)} }^{z_k^{(1)}  v  }\ri),
\eea
where, by definition, the regularization is made by subtraction of the singular part in the $z$--coordinate of the antiderivative.

\paragraph{Variations of  divisor $D$.}
We now consider a derivative $\pa_{\l_j}$. Since $v$ and the spectral curve are independent of $\lambda_j$'s,  only the first term in \eqref{resid} gives a nonzero contribution.

Using the Riemann bilinear relations we find  
\be
\Theta(\pa_{\l_j}) =\le(      \le(\res{\infty^{(1)} }+\res{\infty^{(2)} }\ri) 
+ 
\sum_{\ell}  \le(\res{\l_\ell^{(1)} }+\res{\l_\ell^{(2)} }\ri) \ri)
 v_0 \int\frac {\pa v_0}{\pa \l_j}.
\ee
Now observe that for $\ell\neq j$  and for the residues at infinity  the integrand  is locally analytic and hence the differential $v_0\int \pa_{\l_j} v_0$ has a simple pole; we can pull the derivative outside of the integration because the regularization depends on $\l_\ell$ but not on $\l_j$.
Thus we have 
\be
\Theta(\pa_{\l_j}) = \frac  1 2  \le(\int_{\infty^{(2)} }^{\infty^{(1)} } + \sum_{\ell=1\atop \ell\neq j}^g \int_{\l_\ell^{(2)} }^{\l_\ell^{(1)} }\ri) \!\!\frac {\pa v_0}{\pa \l_j} 
+\le( \res{\l_j^{(1)} }+\res{\l_j^{(2)} }\ri)v_0 \int\frac {\pa v_0}{\pa \l_j}.
\ee
 We are left with the contribution of $\ell=j$:
 \bea
  \res{(\l_j, \pm \mu_j)} v_0 \int\frac {\pa v_0}{\pa \l_j}.
  \label{redv0}
 \eea
  The local behaviour of the indefinite integral is (near $\l_j^{(1)} $):
\be
  \int^p \pa_{\l_j} v_0 = \int ^p\frac \pa{\pa \l_j} \le(
  \frac {1} {2(z-\l_j)} + \mathcal O(1)
  \ri)\d z= \int ^p \le(
  \frac {1} {2(z-\l_j)^2} + \mathcal O(1)
  \ri)\d z= \frac {-1}{2(z-\l_j)} + \mathcal O(1).
\ee
Therefore, the local behaviour of the function we are taking the residue of in \eqref{redv0}  is 
 \be
  v_0 \int\frac {\pa v_0}{\pa \l_j} = \le(\frac {\pm 1}{2(z-\l_j)} + A^{(1,2)} _j + ...\ri)\le(\frac {\mp 1}{2(z-\l_j)} + C^{(1,2)} _j + \dots \ri)
 \ee
 where $A_j^{(1,2)}  = v_0^{reg}(\l_j^{(1,2)} ) = \pm v_0^{reg}(\l_j^{(1)} )$. This  means that the result is $\frac{A_j^{(2)}  - A_j^{(1)}  +  C_j^{(1)}  - C_j^{(2)} }2$.  Now, by definition of regularization:
\be
\frac{C_j^{(1)}  - C_j^{(2)} }2 =\frac 1 2 \, \slint_{\l_j^{(2)} }^{\l_j^{(1)} } \frac {\pa v_0}{\pa \l _j}.
\ee
We then observe that 
\be
\frac 1 2 \,\slint_{\l_j^{(2)} }^{\l_j^{(1)} } \frac {\pa v_0}{\pa \l _j}
=
\frac 1 2 \frac {\pa }{\pa \l _j}\,
 \slint_{\l_j^{(2)} }^{\l_j^{(1)} } v_0 - \frac 1 2  v_0^{reg}\bigg|_{\l_j^{(2)} }^{\l_j^{(1)} } 
 = 
\frac 1 2 \frac \pa{\pa \l_j} \,\slint_{\l_j^{(2)} }^{\l_j^{(1)} } v_0 +\frac{ A_j^{(2)}  -  A_j^{(1)}   }2.
\ee
Therefore we compute
\bea
\res{\l_j^{(1,2)}  }\left( v_0 \int \pa_{\l_j} v_0\right) &=  \frac {A_j^{(2)}  - A_j^{(1)}  + C_j^{(1)}  - C_j^{(2)} }2 = 
\frac 1 2 \,\slint _{\l_j^{(2)} }^{\l_j^{(1)} } \frac \pa{\pa \l_j}v_0 
+ \frac{ A_j^{(2)}  - A_j^{(1)} }2 =
\frac 1 2 \frac \pa{\pa \l_j}\,\slint _{\l_j^{(2)} }^{\l_j^{(1)} } v_0 - 2A_j^{(1)}  .
\label{143}
\eea
It remains to compute $A_j^{(1)} $: 
 from the definition of $v_0$ \eqref{v0} it follows promptly
\be
A_j^{(1)}  =  \frac{Q_1^{reg}(\l_j)}{2\mu_j} +  \frac 1 2 \pa_{\l_j}\ln \mu_j
=
\frac 1{2\mu_j} \le(\sum_{\ell\neq j} \frac {\mu_{\ell}}{\l_j-\l_\ell} + \sum_{k} \frac {\gamma_k} {\l_j-z_k}\ri)  +  \frac 1 2 \pa_{\l_j}\ln \mu_j
\ee
Using the Bethe equations \eqref{Betther} one finds then 
\be
A_j^{(1)}  = \frac{1}2 \le(\sum_{\ell \neq j} \frac 1{\l_j-\l_\ell} - \sum_k \frac 1{\l_j-z_k} \ri) 
+
   \frac 1 2 \pa_{\l_j}\ln \mu_j.
\ee
Inserting this expression into \eqref{143} we then obtain:
\bea
\res{\l_j^{(1,2)}  } v_0 \int \pa_{\l_j} v_0
&=\frac 1 2 \frac \pa{\pa \l_j}\slint _{\l_j^{(2)} }^{\l_j^{(1)} } v_0  -  \le(\sum_{\ell \neq j} \frac 1{\l_j-\l_\ell} - \sum_k \frac 1{\l_j-z_k} \ri) -  \pa_{\l_j}\ln \mu_j
\eea
where, again the regularization so far is made in the $z$--coordinate.

Thus 
\be
\sum_{j=1}^{g} \Bil{v_0}{\frac {\pa v_0}{\pa \l_j}} \d \l_j 
=
 \frac 1 2  \d_{\l}
\le( \slint_{\infty^{(2)} }^{\infty^{(1)} }+   \sum_{j=1}^g \slint_{\l_j^{(2)} }^{\l_j^{(1)} } 
\ri)v_0 
+ \d_\l \ln \frac { \prod_{a,k} (\l_a-z_k)}{\prod \mu_j \prod_{a<b} (\l_a-\l_b)}
\ee
This concludes the proof. \QED

\appendix
\section{ Kirillov-Kostant symplectic potential }

\subsection{ Szeg\"o kernel and its variations}

Here we list the necessary information about Szeg\"o kernel and its variations. 
For a Riemann surface of genus $g$ denote the Abel map by $\Acal(x)$, introduce holomorphic differentials 
$u_j$ normalized via $\int_{a_j} u_k=\delta_{jk}$ and the prime-form $E(x,y)$. Let $q\in \C^g$.
The Szeg\"o kernel $S_q$ is then given by
\be
S_q(x,y)= \f{\Theta(\Acal(x)-\Acal(y)+q)}{\Theta(q) E(x,y)}
\la{Sz}
\ee
The Szeg\"o kernel has the following properties. First, it has simple pole on the diagonal of the form:
\be\label{ee:SzExp}
S_q(x,y)= \left(\f{1}{\xi(x)-\xi(y)} +\mathcal{O}(1) \right)\sqrt{{\rm d}\xi(x)}\sqrt{{\rm d}\xi(y)}
\ee
where $\xi$ is a local coordinate near the diagonal.
Second, it has the following periodicity properties: $S_q(x,y)$ remains invariant (up to a sign) when $x$ or $y$ are analytically continued along any $a$-cycle $a_j$; under analytical continuation along $b_j$ one has
$$
S_q(x+b_j,y)= e^{-2\pi i q_j}S_q(x,y)\;, \hskip0.7cm
S_q(x,y+b_j)= e^{2\pi i q_j}S_q(x,y) \;\;\; 
$$

The Szeg\"o kernel satisfies the following identity due to Fay \cite{Fay73}:
\be
\label{Fayid}
S_q(x,y) S_q(y,x)= -B(x,y)- \sum_{\a,\b=1}^g \p_\a\p_\b \log \Theta_q(0) u_\a(x) u_\b(y)
\ee
where $B(x,y)={\rm d}_x{\rm d}_y\log E(x,y)$ is the canonical bimeromorphic differential.

The Szeg\"o kernel depends on the moduli of the Riemann surface $\CC$ (we consider here the moduli space 
of hyperelliptic curves of genus $g$ defined by (\ref{spcurve})) and on the vector $q$, which defines a point of the Jacobian of $\CC$. Variational formula for Szeg\"o kernel  on the space $\Acal^{\bf r}$ can be conveniently written in terms of coordinate system $(A_j,q_j)$ where $A_j=\int_{a_j}v$ and $q_j$ are components of vector $q$. The moduli of the curve $\CC$ (for fixed ${\bf r}$ and $z_j$) depend (locally) only on the periods $A_j$.

The variational formulas are given in the next proposition.

\begin{proposition} The following variational formulas hold: 
\be
\f{\p}{\p A_j}S_q(x,y)=-\f{2\pi i }{4}\sum_{i=1}^{2g+2} \f{u_j}{{\rm d} \ln (v/{\rm d} z)}(x_i)\,\res{x_i}\f{W_t [S_q(x,t),\, S_q(t,y)]}{v(t)}
\la{SI}
\ee
\be
\f{\p}{\p q_\gamma}S_q(x,y)=-\int_{t\in a_\gamma} S_q(x,t) S_q(t,y)
\la{Sqv}
\ee
where $W_t (f,g)$ denotes Wronskian of two functions $f(t)$ and $g(t)$
 .
\end{proposition}
{\it Proof.} The formula (\ref{Sqv}) was proved in  \cite{KalKor} (Prop.1).

The formula (\ref{SI}) follows from two results. The first is Theorem 2 of \cite{KalKor}
where the variational formulas for $S_q$ on moduli spaces of meromorphic abelian differentials were derived. Then (\ref{SI}) can be obtained from the formula of  \cite{KalKor} via chain rule, following verbatim the proof of formula (3.10) from
\cite{BK_Hitchin} where the variational formulas for Abelian differentials were derived.

\QED

We shall need the following lemma which is valid for any $n$-sheeted cover of $\mathbb{P}^1 $; this statement is equivalent to relations (4.12), (4.13) from  \cite{Annalen}.

\begin{lemma}
\label{lemma51}
Let $\CC$ be an $n$-sheeted cover of $\mathbb{P}^1 $ with projection $f:\CC\to \mathbb P^1$. Then 
the following identity holds:
\be
\sum_{i=1}^n S_q(x,t^{(i)})S_q(t^{(i)},y)= S_q(x,y)\left(\f{1}{f(x)-t}-\f{1}{f(y)-t}\right){\rm d}t 
\la{Sasso}
\ee
\end{lemma}
{\it Proof.} For completeness here we give a short independent proof of this fact.  The l.h.s. of 
(\ref{Sasso}) is a 1-form in $t$   depending only on the point of the base. It has simple poles at 
$z=f(x)$ and $z=f(y)$. The coefficient depends on $x$ and $y$
and must coincide with $S_q(x,y)$ due to the holonomy properties of $S_q$.
\QED

We shall also use the following notations. First, introduce the system of distinguished local coordinates on $\CC$.
Near a branch point $x_j$ it is given by $\zeta_j(z)=\sqrt{z-x_j}$. Near $\infty^{(1,2)}$ the distinguished coordinate is chosen to be $\xi(z)=1/z$. Finally, near any other point with projection $z_0$ on $z$-plane the distinguished coordinate is $z-z_0$.
Now we define the following:
\be
S_q(x,\infty^{(j)})=\frac{S_q(x,y)}{{\sqrt{{\rm d}(z^{-1}(y))}}}\Big|_{y=\infty^{(j)}}\;,\hskip0.7cm j=1,2
\la{Sinf}
\ee
\be
S_q(x,x_k)=\frac{S_q(x,y)}{{\sqrt{{\rm d}\sqrt{f(y)-x_j}}}}\Big|_{y=x_k}\hskip0.7cm k=1,\dots,2g+2
\la{Sbp}
\ee

Using these notations we get from  (\ref{Sasso}) in the limit $t\to \infty$:
\be
\sum_{i=1}^n S_q(x,\infty^{(i)})S_q(\infty^{(i)},y)= S_q(x,y)(f(x)-f(y))
\la{SSS}
\ee

\subsection{Eigenvectors of rational matrix functions via Szeg\"o kernel on spectral curve}

The Szeg\"o kernel can be conveniently used to construct eigenvectors of the rational matrix-valued function $R(z)$ (\ref{abcd}).
The construction is parallel to the one used in \cite{Annalen} to construct solutions of matrix Riemann-Hilbert problems. 

\begin{proposition}
Consider a pair $(\CC,q)$
where $q\in \C^g$ and $\CC$ is the spectral curve given by equation 
\be
 \mu^2=C\frac{\prod_{j=1}^{2g+2}(z-x_j)}{\prod_{j=1}^{g+2} (z-z_j)^2}
 \la{spc}
\ee
such that for the differential $v=\mu {\rm d} z$ we have
\be
\res{z_j^{(1,2)}}v =\pm r_j\;, \hskip0.7cm
\res{\infty^{(1,2)} } v =\pm r_\infty
\la{resco}
\ee
where $r_1,\dots, r_{2g+2}, r_\infty$ are some constants.
Consider the  divisor $D=\l_1^{(1)} +\dots+\l_g^{(1)}$ defined by
\be
\Acal_{\infty^{(2)}}(D)=-q +K^{\infty^{(2)}}
\la{defD}
\ee

Consider the canonical polygon $\tilde{\CC}$ invariant under the hyperelliptic involution $\nu$.
Define the following column-vector for $x\in \tilde{\CC}$:
\be
\psi(x)=\f{1}{\sqrt{{\rm d} f(x)}} \left(\ba{c} S_{q}(x,\infty^{(1)}) \\  S_{q}(x,\infty^{(2)} ) \ea\right)
\
\ee
and the $2\times 2$ matrix on $\tilde{\CC}$
\be
\Psi(x) =\left(\psi(x), \psi(x^*)\right)=\f{1}{\sqrt{{\rm d} f(x)}} \left(\ba{cc} S_{q}(x,\infty^{(1)}) & S_{q}(x^*,\infty^{(1)})\\  S_{q}(x,\infty^{(2)} )  & S_{q}(x^*,\infty^{(2)})  \ea\right)
\la{Psi}
\ee
Then the matrix $R$ defined by
\be
R(x) {\rm d} z (x)
=\Psi(x) \left(\ba{cc}  v & 0 \\ 0 & -v \ea\right) \Psi^{-1}(x)
\la{RPsi}
\ee
is a rational matrix invariant under the transformation $x\to x^\nu$ 
i.e. it depends only on $z$, and 
\begin{enumerate}
\item
$R(z)$  has simple poles only at $z_j$:
\be
R(z)=\sum_{j=1}^{g+2} \frac {R_j}{z-z_j}\;
\la{Rx}
\ee
\item
The eigenvalues of $R_j$ are equal to $\pm r_j$ and  
$$R_\infty:=-\sum_{j=1}^{g+2} R_j=   \left(\ba{cc} r_\infty & 0 \\ 0 & -r_\infty\ea\right)\;.$$
\item
 $\CC$ coincides with the spectral curve  ${\rm det} (R(z)-\mu I)=0$.
 \item
 The matrix entry $R_{21}(z)$ has on $\CC$ exactly $g$ zeros situated at 
 $\lambda_j\equiv f(\l_j^{(1)})$. 

\item
The points $\l_j^{(1)}=(\lambda_j,\mu_j)$ are such that $\mu_j=R_{11}(\lambda_j)$ i.e. the divisor $D$ defined by  
(\ref{defD}) is given by $D=\sum_{j=1}^g (\lambda_j, R_{11}(\lambda_j))$. Another characterization of the divisor is that it consists of the points  of the spectral curve above the finite part of the plane  where the eigenvector is proportional to the vector $(1,0)$.
\end{enumerate}

\end{proposition}
{\it Proof.} 
To prove that the matrix $R$ is invariant under the involution $x\to x^\nu$ we notice that $\Psi(x^\nu)=\Psi (x)\sigma_1$
and $v(x^\nu)=-v(x)$. Therefore, $ R(x^\nu)=R(x)$.

The eigenvalues of the matrices $R_j$ coincide with residues of $v$ at the points $z_j^{(1,2)}$ which are equal to
$\pm r_j$ due to (\ref{resco}). 

According to  \cite{Annalen} (see p.350, and formula (4.14)), Fay's identities imply that ${\rm det}\Psi=\pm 1$ and
\be
\Psi^{-1}(x)=
%
\f{1}{\sqrt{{\rm d} f(x)}} \left(\ba{cc} S_{q}(\infty^{(1)},x) & S_{q}(\infty^{(2)},x)\\  
S_{q}(\infty^{(1)},x^\nu )  & S_{q}(\infty^{(2)},x^\nu)  \ea\right)
\la{Psiin}
\ee

The matrix element $R_{21}$  is a rational function of $z=f(x)$  given by
$$
R_{21}(z) = 2\f{v(x)}{({\rm d} f(x))^2} S_{q}(x,\infty^{(2)} ) S_{q}(x^\nu,\infty^{(2)})
$$
which vanishes at the points of divisor $D+ D^\nu$ due to (\ref{defD}). Equivalently, 
it means that $R_{21}$, considered as function of $z$, vanishes at the points of $\pi(D)$.

Finally, the eigenvalues of $R_j$ from (\ref{Rx}) are equal to $\pm$ residues of $v$ at $z_j^{(1,2)}$, i.e. $\pm r_j$.

\QED

The expression (\ref{Psi}) can be alternatively written as follows:
\be
\Psi_{\a\b}(z)=\f{1}{\sqrt{{\rm d} f(x)}} S_{q}(z^{(\b)},\infty^{(\a)}) 
\la{Psi1}
\ee
and (\ref{Psiin}) as
\be
(\Psi^{-1})_{\a\b}(z)=\f{ 1}{\sqrt{{\rm d} f(x)}} S_{q}(\infty^{(\b)},z^{(\a)}) 
\la{psim1}
\ee

\begin{corollary}
The matrices $R_j$ can be diagonalized as follows:
\be
R_j= G_j \left(\ba{cc} r_j & 0 \\ 0 & -r_j\ea\right) G_j^{-1}
\la{Rdia}
\ee
where the formulas for $G_j$ are obtained for (\ref{Psi}), (\ref{Psi1}):

\be
(G_j)_{\a\b}=S_{q}(z_j^{(\b)},\infty^{(\a)}) 
\la{GS}
\ee
such that
\be
(G_j^{-1})_{\a\b}=S_{q}(\infty^{(\b)},z_j^{(\a)}) 
\la{GSi}
\ee
\end{corollary}

\subsection{Kirillov-Kostant potential}

The Kirillov-Kostant symplectic form looks as follows in terms of eigenvector matrices $G_j$ \cite{BBT}:
\be
\omega_{KK}=-\sum_{j=1}^n {\rm tr}\; L_j   G_j^{-1} \d G_j \wedge    G_j^{-1} \d G_j 
\la{oKK}
\ee
The natural choice of symplectic potential $\theta_{KK}$ such that ${\rm d}\theta_{KK}=\omega_{KK}$ is
\be
\theta_{KK}=\sum_{j=1}^n {\rm tr}\; L_j   G_j^{-1} \d G_j 
\la{thKK}
\ee
or, alternatively,

\be\label{thKK2}
\theta_{KK} = 
\sum_{j=1}^{g+2} \res{z_j} \tr\bigg(\widehat v \Psi^{-1} \d \Psi\bigg)=
\sum_{j=1}^{g+2} \res{z_j} \tr\bigg(\widehat \mu \Psi^{-1} \d \Psi\bigg) {\rm d}z
\ee
where
\be
\widehat v=\widehat \mu  {\rm d}z=\left(\ba{cc}  v & 0 \\ 0 & -v \ea\right)
\ee

\begin{theorem}\la{TKK}
The form $\theta_{KK}$ in $(A_j,q_j)$-coordinates looks as follows:
\be
\theta_{KK}=\sum_{j=1}^g A_j \delta q_j
\ee
where $A_j=\int_{a_j} v$ are $a$-period of $v$ (defined up to an integer linear combination of Casimirs $r_j$).
\end{theorem}
{\it Proof.}
Consider first the contribution of $\delta q_j$ in (\ref{thKK}). First, using (\ref{Psi1}), (\ref{psim1})
and variational formula (\ref{Sqv}), we get

\be
\bigg(\Psi^{-1}(z) \pa_{q_j} \Psi (z)\bigg)_{\a\b}{\rm d} z  =-\sum_{\gamma=1}^2  \oint_{x\in a_j} S_q( z^{(\b)} ,x)S_q(x,\infty^{(\g)}) S_q(\infty^{(\g)}, z^{(\a)} )
\ee
or, using Lemma \ref{lemma51},
\be
\bigg(\Psi^{-1}(z) \pa_{q_j} \Psi (z)\bigg)_{\a\b}{\rm d} z  = - \oint_{x\in a_j} S_q( z^{(\b)} ,x)S_q(x, z^{(\a)} ) (f(x) -z)\;.
\ee
Recall  that $\mu(z)$ is a meromorphic function with simple poles at the points $p\to z_j^{(1,2)}$  with singular parts  $\mu=  \frac{\pm r_j}{z-z_j}$.

The contribution to $\theta_{KK}$ of $\d q_j$ is therefore given by:
\begin{equation}
\theta_{KK}(\partial_{q_j}) = -
 \sum_{k=1}^{g+2} \sum_{\a=1}^2 (L_k)_{\a\a}\oint_{x\in a_j} \frac{S_q( p ,x)S_q(x, p )}{{\rm d} f(p)} \bigg|_{p=z_k^{(\a)}} (f(x) -f(p))
\end{equation}
 which, using  \eqref{thKK2}, gives
\begin{equation}
 -
 \sum_{k=1}^{g+2} \sum_{\a=1}^2 \res{p=z_k^{(\a)}} \mu(p)\oint_{x \in a_j} 
 S_q( p ,x)S_q(x, p ) (f(x) -f(p)).
\end{equation}
 The integration contours $a_j$ in the $x$--variable can be chosen so as not to intersect the integration contours for the residues in the $p$--variable and hence the integrand is regular. Thus we can interchange the order of integrations:
 \be
\theta_{KK}(\partial_{q_j}) =-\oint_{x \in a_j}  \sum_{k=1}^{g+2} \sum_{\a=1}^2 \res{p=z_k^{(\a)}} \mu(p)S_q( p ,x)S_q(x , p ) (f(x) -f(p)) 
\la{thkk}
\ee
The sum over the residues  is the sum over all poles above the $z_j$'s of the differential (in the $p$ variable)
\be
\mu(p)S_q( p ,x)S_q(x, p ) (f(x) -f(p)) \;.
\la{muS}
\ee
This differential does not have a pole at $p=\infty^{(1,2)}$ because the eigenvalue, $\mu(p)$ of $R$ has a simple zero, which cancels the pole of $f(p)$. Moreover it 
 has an  additional  simple pole at $p=x$, with residue 
\be
\res{p=x} \mu (p)S_q( p ,x)S_q(x, p ) (f(x) -f(p))  = \mu (x) {\rm d}f(x).
\ee
 Thus we can use  residue theorem and rewrite (\ref{thkk}) as follows:
\be
\theta_{KK}(\partial_{q_j}) =\oint_{x\in a_j}\res{p=x} \mu(p)S_q( p ,x)S_q(x, p ) (f(x) -f(p)) 
 \ee
Furthermore, using Fay's identity (\ref{Fayid}), and using the fact that $B(p,x)$ behaves on the diagonal as $(f(p)-f(x))^{-2}{\rm d}f(p){\rm d}f(x)$, we get
\be
\res{p=x} \mu (p)S_q( p ,x )S_q(x, p ) (f(x) -f(p))=\mu (x){\rm d} f(x)
\ee
and, therefore,
\be
\theta_{KK}(\partial_{q_j})=\oint_{a_j}v=A_j
\ee

Consider now the contribution of $\delta A_j$ to $\theta_{KK}$.  
Remind that $\Psi(\l_j)=G_j$ while ${\rm res}|_{z_j} \hat{v} =L_j$. Now, using variational formulas (\ref{SI}) we have
$$
\left(\Psi^{-1}\p_{A_j}\Psi\right)_{\a\a}=
$$
$$
=\sum_{\g=1}^2 \frac{S_q(\infty^{(\g}),p)}{{\rm d}f(p)}\left(\f{\pi i}{2}\sum_{x_i} u_j(x_i)
{\rm res} |_{x=x_i}\f{W_x[S_q(p,x), S_q(x,\infty^{(\g)})]}{v(x)}\right)\bigg|_{p=z_j^{(\a)}}
$$
where
\be
u_j(x_i)= \frac{u_j}{{\rm d}\ln (v/\d\xi)}(x_i)
\ee

Therefore, using (\ref{SSS}) we have
$$
\theta_{KK}(\pa_{A_j}) =
\sum_{k=1}^{g+2} \sum_{\a=1}^2 (L_k)_{\a\a} \left(\Psi^{-1}\p_{A_j}\Psi\right)_{\a\a}=
 $$
\be
=
-\f{\pi i}{2}\sum_{\ell=1}^{g+2} \sum_{\a=1}^2 \res{}\big|_{p=z_\ell^{(\a)}} \le(\sum_{x_i\in b.pts} u_j(x_i) \res{t=x_i}\frac {W_x\bigg[S_q(p,x),\; S_q(x,p)(f(x)-f(p))\bigg]}{\mu(x){\rm d} f(x) }\mu(p)\ri)
\label{333}
\ee
or, since the branch points  are simple,
\be
\theta_{KK}(\pa_{A_j})=
-\f{\pi i}{2}\sum_{\ell=1}^{g+2} \sum_{\a=1}^2 \res{p=z_\ell^{(\a)}} \le(\sum_{x_i\in b.pts}  \res{t=x_i}  \frac {u_j(x)}{{\rm d}\ln \mu(x))} \frac {W_x\bigg[S_q(p,x),\; S_q(x,p)(f(x)-f(p))\bigg]}{\mu (x){\rm d} f(x) }\mu (p)\ri)
\ee
Once again we can swap the order of residues because the branch-points are away from the points $z_j^{(\a)}$. 
One can verify that the differential of $p$ in the inmost residue has poles only above the $z_j$'s and at $p=t$,  with no pole at $\infty$, for the same reason as in equation \eqref{muS}.
 To compute the residue at $p=x$, we use \eqref{ee:SzExp} that specifies the behavior of $S_{q}(p,x)$ for $p\sim x$:
\bea
&\res{p=x}\frac{W_x\bigg(S_q(p,x),\,S_q(x,p)[f(x)-f(p)] \bigg)}{\mu (x) {\rm d} f(x)}\mu (p)\nonumber\\
&=\res{p=x} 
\frac{\partial_x S_q(p,x)S_q(x,p)(f(x)-f(p))-S_q(p,x)\partial_x S_q(x,p)(f(x)-f(p))-S_q(p,x)S_q(x,p)}{{\mu (x){\rm d}f(x)}}\mu (p)
\nonumber\\
&=\res{p=x}
\frac{1}{(f(p)-f(x))^2}\frac{\mu  (p)}{\mu  (x)} 
=d\log \mu (x).
\eea
Therefore,
$$
 \theta_{KK}(\pa_{A_j})
 =-\f{\pi i}{2} \sum_{x_i} \res{x=x_i} \f{u_j(x)}{{\rm d}\ln \mu(x)} {\rm d}\ln \mu (x)=0
 $$

\QED

\section{Bergman tau-function}
\label{secBerg}
Here we summarize the key facts from the theory of  Bergman tau-function (see \cite{Annalen,JDG,KalKor} and the review \cite{KorRev}). 

First, let us write the function  (\ref{eq:Q0exp}) as follows:
\be
Q_0(z) =C_0 \frac {P(z)}{R^2(z)}
\la{Q01}
\ee
where 
$$P(z)=\prod_{j=1}^{2g+2}(x-x_j)\;,\hskip0.7cm R(z)=\prod_{j=1}^{g+2}(z-z_j)$$ 
and
$$
C_0=\sum_{j=1}^{g+2} r_j^2
$$

\subsection{Definition of  $\tau_B(Q_0)$ and its properties}

Here we remind the definition of  the Bergman tau-function on the space of quadratic differentials 
(see \cite{BK1}, where the tau-function needed here is denoted by $\tau_+$, and the review \cite{KorRev}).
The constructions of \cite{BK1} are adjusted to the case of genus zero base curve (then in the notations of 
\cite{BK1} we have $g_-=g$).



%


In the genus zero case the Bergman tau-function  $\tau_B$ is defined by the following equations with respect to the periods of $v$ along canonical cycles on $\CC$  (see \cite{KalKor,BK1}):
\be
\frac{\d\log\tau_B}{\d A_i }=-\frac{1}{24\pi i }\int_{b_i} \frac{\Scal\left(\int^x v, z(x)\right)(d z)^2}{v} \;,\hskip0.7cm
\frac{\d\log\tau_B}{\d B_i}=\frac{1}{24\pi i }\int_{a_i} \frac{\Scal\left(\int^x v, z(x)\right)(d z)^2}{v}
\la{defBerg}
\ee
for $i=1,\dots,g$. Here $z$ is the global coordinate on $CP1$ and $\Scal(\cdot,\cdot)$ is the Schwarzian derivative
(notice that in $z$-coordinate the Bergman projective connection $S_B$ on the Riemann sphere is identically zero).

Therefore,  the differential of $\log\tau_B$ on the symplectic leaf $\Qcal_{g,n}[{\bf r}]$ is given by
the following expression:
\be
\d\log\tau_B=\frac{1}{24\pi i }\sum_{j=1}^{g} \left[\left(\int_{a_i}  \frac{\Scal\left(\int^x v, z(x)\right) (d z)^2 }{v} \right)   \d B_j -  \left(\int_{b_i }\frac{\Scal\left(\int^x v, z(x)\right) (d z)^2 }{v} \right)   \d A_j  \right]\;.
\la{dlogtau1}
\ee

\subsection{Explicit formula for $\tau(Q_0)$}

Let us introduce the following regularized integrals on $\Ch$ of $v=\sqrt{Q_0}{\rm d} z$:
\be
{\rm reg}\int_{x_1}^{z_k} v:=\lim_{x\to z_k}\left\{ \int_{x_1}^x v - r_k \ln (z(x)-z_k)           \right\}
\la{regint}
\ee

The explicit formula for $\tau_B(Q_0)$ is then given by the following proposition:
\begin{proposition}
Choose the contours $l_j$ connecting $x_1$ with $z_j$ such that they lie entirely inside of the fundamental polygon 
$\Ch$. 

\be
\tau_B=
\exp\left\{-\f{1}{6} \sum_{k=1}^{g+2} \f{1}{r_k} {\rm reg} \int_{x_1}^{z_k} v  \right\}
 \prod_{j<k}(x_j-x_k)^{5/144}  \prod_{j<k}(z_j-z_k)^{1/6} \prod_{j,k}(x_j-z_k)^{-7/72}
 \la{tauf}
\ee
\la{tauprop}
\end{proposition}

{\it Proof.}
Denote by 
\be
E(z,w)=\f{z-w}{\sqrt{dz}\sqrt{dw}}
\ee
 the prime-form on $\mathbb P^1$. In terms of the  prime-form the Bergman tau-function is given by the following expression valid for the base curve of genus zero  in terms of the divisor $(Q_0)=\sum d_j p_j$:
\cite{JDG,contemp,BK1,KalKor}:
\be
\tau_B=
\left(\f{\qd_0(x)}{\prod_{j}E^{d_i}(x,p_j)}\right)^{-1/6} 
\prod_{i<j} E(p_i,p_j)^{\f{d_i d_j}{24}}\, .
\la{taupint}
\ee

The  prime-forms in (\ref{taupint})  are evaluated at the points $p_i$ as follows:
\be
E(x,q_i)=\lim_{y\to p_i} E(x,y) \sqrt{d\zeta_i(y)},
\la{defEpi}
\ee
\be
 E(p_i,p_j)=\lim_{x\to p_j, y\to p_i} E(x,y) \sqrt{d\zeta_i(y)} \sqrt{d\zeta_j(x)}
\la{defEppi}
\ee
where $\zeta$ are distinguished local coordinates near points $p_j$.

Let us apply (\ref{taupint}) to our case, when 
$$(\qd_0)=\sum_{k=1}^{2g+2}x_j -2 \sum_{k=1}^{g+2} z_k\;.$$
and
\be
\tau_B=\left(\f{\qd_0(x) \prod_{j}E^2(x,z_k) }{\prod_{j}E(x,x_j)}\right)^{-1/6}
\left( \f{\prod_{i<j} E(x_i,x_j)  \prod_{k<l} E^4(z_k,z_l)}{ \prod_{k,l} E^2(x_k,z_l) }  \right)^{{1}/{24}}\;.
\la{tau1}
\ee

Locally near $z_k^{(1)}$ we have 
\be
v\sim \f{r_k}{z-z_k}d z+\dots
\ee
and the distinguished local coordinate $\zeta_k$ near $z_k$ is given by
\be
\zeta_k(x)=\exp\left\{ \f{1}{r_k}\int_{x_1}^z v\right\}\;.
\la{zetak}
\ee
such that
\be
\f{d \zeta_k}{d z}\big|_{z_k}= \exp\left\{ \f{1}{r_k} {\rm reg} \int_{x_1}^{z_k} v\right\}
\ee
where the regularized integral is given by (\ref{regint}).

The total power of $d\zeta_j/d z (z_j)$ in (\ref{taupint}) is 
\be
-\f{1}{6}+\f{1}{12} (g+1) -\f{1}{24}(2g+2)=-\f{1}{6}
\ee
which gives the second term in (\ref{tauf}).

The distinguished local coordinates $\xi_j$ near $x_j$ are given by:
\be
\xi_j(x)=\left(\int_{x_j}^z v\right)^{2/3}
\ee
Since locally, near the branch point $x_j$, we have
\be
v\sim a_j(z-x_j)^{1/2}\;,
\ee
we get
\be
\xi_j(z)\sim a_j^{2/3} (z-x_j)
\ee
and
\be
\f{ d\xi_j(z)}{d z} (x_j)= a_j^{2/3}
\ee
where
$$
a_j=C_0^{1/2}\frac{\prod_{k\neq j} (x_k-x_j)^{1/2}}{\prod_{k} (x_j-z_k)}
$$

The total power of $d\xi_j/d z (x_j)$ in  (\ref{tau1}) is
\be
\f{1}{12}+\f{1}{2} \f{1}{24}(2g+1) -\f{1}{24}(g+2) =\f{1}{48}
\ee
Therefore, the total power of $a_j$ is 
$
\f{1}{72}
$\;.

The total power of $C_0$ comes from $Q$ and all $a_j$ which gives
$$
-\f{1}{6}+ \f{1}{2}\f{1}{72} (2g+2)  =\f{g-11}{72}
$$
However, we don't include the $C_0$ multiplier into $\tau_B$ since the latter is defined up to a constant which might depend on residues.

Let us now compute the powers of $x_j-x_k$, $z_j-x_k$ and $z_j-z_k$.

The product of $z_j-z_k$ comes only from $E(z_j,z_k)$, thus equals to $1/6$.

The product of $x_j-z_k$ comes from $E(x_j,z_k)$ (gives $-1/12$) and from product of $a_j$ (gives $-1/72$), and in total we get 
$$
-\f{1}{12}-\f{1}{72}= -\f{7}{72}
$$
Finally, the product of $x_j-x_k$ comes from $E(x_j,x_k)$ (gives $1/48$) and product of $a_j$ (gives $1/72$). In total we get $5/144$.
\QED

\providecommand{\href}[2]{#2}\begingroup\raggedright\endgroup

%
%

%
%
%
%

\begin{thebibliography}{10}

\bibitem{AlekMal2}
A.~Y. {Alekseev} and A.~Z. {Malkin}, \emph{{The hyperbolic moduli space of flat
  connections and the isomorphism of symplectic multiplicity spaces}},
  {\emph{eprint arXiv:dg-ga/9603017} (1996) dg}
  [\href{https://arxiv.org/abs/dg-ga/9603017}{{\ttfamily dg-ga/9603017}}].

\bibitem{AlBrid}
D.~Allegretti and T.~Bridgeland, \emph{{The monodromy of meromorphic projective
  structures}}, \href{https://doi.org/10.1090/tran/8093}{\emph{Trans. Am. Math.
  Soc.} {\bfseries 373} (2020) 6321}.

\bibitem{Alleg}
D.~G.~L. Allegretti, \emph{{Voros symbols as cluster coordinates}},
  \href{https://doi.org/10.1112/topo.12106}{\emph{J. Topol.} {\bfseries 12}
  (2019) 1031} [\href{https://arxiv.org/abs/1802.05479}{{\ttfamily
  1802.05479}}].

\bibitem{BBT}
O.~Babelon, D.~Bernard and M.~Talon, \emph{Introduction to Classical Integrable
  Systems}, Cambridge Monographs on Mathematical Physics. Cambridge University
  Press, 2003,
  \href{https://doi.org/10.1017/CBO9780511535024}{10.1017/CBO9780511535024}.

\bibitem{BGG2021}
M.~Bershtein, P.~Gavrylenko and A.~Grassi, \emph{{Quantum spectral problems and
  isomonodromic deformations}},
  \href{https://arxiv.org/abs/2105.00985}{{\ttfamily 2105.00985}}.

\bibitem{BK_Hitchin}
M.~Bertola and D.~Korotkin, \emph{{Spaces of Abelian Differentials and
  Hitchin’s Spectral Covers}},
  \href{https://doi.org/10.1093/imrn/rnz142}{\emph{International Mathematics
  Research Notices} (2019) }.

\bibitem{BK1}
M.~Bertola and D.~Korotkin, \emph{{Hodge and Prym Tau Functions, Strebel
  Differentials and Combinatorial Model of ${\mathcal {M}}_{g,n}$}},
  \href{https://doi.org/10.1007/s00220-020-03819-9}{\emph{Commun. Math. Phys.}
  {\bfseries 378} (2020) 1279}
  [\href{https://arxiv.org/abs/1804.02495}{{\ttfamily 1804.02495}}].

\bibitem{BK_TMF}
M.~Bertola and D.~A. Korotkin, \emph{{WKB expansion for a~Yang\textendash{}Yang
  generating function and the~Bergman tau function}},
  \href{https://doi.org/10.4213/tmf9834}{\emph{Teor. Mat. Phys.} {\bfseries
  206} (2021) 295}.

\bibitem{BK_JDG}
M.~Bertola and D.~Korotkin, \emph{Extended goldman symplectic structure in
  fock-goncharov coordinates}, {\emph{To appear in J. Diff. Geom.} (2021) }
  [\href{https://arxiv.org/abs/1910.06744}{{\ttfamily 1910.06744}}].

\bibitem{BK2iso}
M.~Bertola and D.~Korotkin, \emph{{Tau-Functions and Monodromy
  Symplectomorphisms}},
  \href{https://doi.org/10.1007/s00220-021-04224-6}{\emph{Commun. Math. Phys.}
  {\bfseries 388} (2021) 245}
  [\href{https://arxiv.org/abs/1910.03370}{{\ttfamily 1910.03370}}].

\bibitem{BLMST2017}
G.~Bonelli, O.~Lisovyy, K.~Maruyoshi, A.~Sciarappa and A.~Tanzini, \emph{{On
  Painlev{\'e}/gauge theory correspondence}},
  \href{https://doi.org/10.1007/s11005-017-0983-6}{\emph{Letters in
  Mathematical Physics} {\bfseries 107} (2017) 2359}
  [\href{https://arxiv.org/abs/1612.06235}{{\ttfamily 1612.06235}}].

\bibitem{Bri}
T.~Bridgeland, \emph{{Riemann-Hilbert problems from Donaldson-Thomas theory}},
  \href{https://doi.org/10.1007/s00222-018-0843-8}{\emph{Invent. Math.}
  {\bfseries 216} (2019) 69}
  [\href{https://arxiv.org/abs/1611.03697}{{\ttfamily 1611.03697}}].

\bibitem{BM}
T.~{Bridgeland} and D.~{Masoero}, \emph{{On the monodromy of the deformed cubic
  oscillator}}, {\emph{arXiv e-prints} (2020) arXiv:2006.10648}
  [\href{https://arxiv.org/abs/2006.10648}{{\ttfamily 2006.10648}}].

\bibitem{BS}
T.~Bridgeland and I.~Smith, \emph{Quadratic differentials as stability
  conditions}, {\emph{Publications math{\'e}matiques de l'IH{\'E}S} {\bfseries
  121} (2015) 155}.

\bibitem{CLT}
I.~Coman, P.~Longhi and J.~Teschner, \emph{{From quantum curves to topological
  string partition functions II}},
  \href{https://arxiv.org/abs/2004.04585}{{\ttfamily 2004.04585}}.

\bibitem{DMDGG2020}
F.~Del~Monte, H.~Desiraju and P.~Gavrylenko, \emph{{Isomonodromic tau functions
  on a torus as Fredholm determinants, and charged partitions}},
  \href{https://arxiv.org/abs/2011.06292}{{\ttfamily 2011.06292}}.

\bibitem{DMDG2022}
F.~Del~Monte, H.~Desiraju and P.~Gavrylenko, \emph{{Monodromy dependence and
  symplectic geometry of isomonodromic tau functions on the torus}},
  \href{https://arxiv.org/abs/2211.01139}{{\ttfamily 2211.01139}}.

\bibitem{DDP}
H.~Dillinger, E.~Delabaere and F.~Pham, \emph{R{\'e}surgence de voros et
  p{\'e}riodes des courbes hyperelliptiques},  in \emph{Annales de l'institut
  Fourier}, vol.~43, pp.~163--199, 1993.

\bibitem{Dub}
B.~Dubrovin, \emph{Integrable systems and riemann surfaces lecture notes}, .

\bibitem{Fay73}
J.~D. Fay, \emph{Theta functions on Riemann surfaces}, vol.~352. Springer,
  2006.

\bibitem{FG}
V.~Fock and A.~Goncharov, \emph{Moduli spaces of local systems and higher
  teichm\"uller theory},
  \href{https://doi.org/10.1007/s10240-006-0039-4}{\emph{Publications
  Math\'ematiques de l'IH\'ES} {\bfseries 103} (2006) 1}.

\bibitem{GMN}
D.~Gaiotto, G.~W. Moore and A.~Neitzke, \emph{Wall-crossing, hitchin systems,
  and the wkb approximation},
  \href{https://doi.org/https://doi.org/10.1016/j.aim.2012.09.027}{\emph{Advances
  in Mathematics} {\bfseries 234} (2013) 239}.

\bibitem{GIL2012}
O.~Gamayun, N.~Iorgov and O.~Lisovyy, \emph{{Conformal field theory of
  Painlev\'e VI}}, \href{https://doi.org/10.1007/JHEP10(2012)038}{\emph{JHEP}
  {\bfseries 10} (2012) 038} [\href{https://arxiv.org/abs/1207.0787}{{\ttfamily
  1207.0787}}].

\bibitem{GL2018}
P.~Gavrylenko and O.~Lisovyy, \emph{{Fredholm Determinant and Nekrasov Sum
  Representations of Isomonodromic Tau Functions}},
  \href{https://doi.org/10.1007/s00220-018-3224-7}{\emph{Commun. Math. Phys.}
  {\bfseries 363} (2018) 1} [\href{https://arxiv.org/abs/1608.00958}{{\ttfamily
  1608.00958}}].

\bibitem{GMS2020}
P.~Gavrylenko, A.~Marshakov and A.~Stoyan, \emph{{Irregular conformal blocks,
  Painlev\'e III and the blow-up equations}},
  \href{https://doi.org/10.1007/JHEP12(2020)125}{\emph{JHEP} {\bfseries 12}
  (2020) 125} [\href{https://arxiv.org/abs/2006.15652}{{\ttfamily
  2006.15652}}].

\bibitem{Gold84}
W.~M. Goldman, \emph{The symplectic nature of fundamental groups of surfaces},
  \href{https://doi.org/https://doi.org/10.1016/0001-8708(84)90040-9}{\emph{Advances
  in Mathematics} {\bfseries 54} (1984) 200}.

\bibitem{Hitchin}
N.~Hitchin, \emph{Frobenius manifolds}, pp.~69--112.
\newblock Springer Netherlands, Dordrecht, 1997.
\newblock 10.1007/978-94-017-1667-3\_3.

\bibitem{ILT2015}
N.~Iorgov, O.~Lisovyy and J.~Teschner, \emph{{Isomonodromic tau-functions from
  Liouville conformal blocks}},
  \href{https://doi.org/10.1007/s00220-014-2245-0}{\emph{Commun. Math. Phys.}
  {\bfseries 336} (2015) 671}
  [\href{https://arxiv.org/abs/1401.6104}{{\ttfamily 1401.6104}}].

\bibitem{ILTy2013}
N.~Iorgov, O.~Lisovyy and Y.~Tykhyy, \emph{{Painlev\'e VI connection problem
  and monodromy of $c=1$ conformal blocks}},
  \href{https://doi.org/10.1007/JHEP12(2013)029}{\emph{JHEP} {\bfseries 12}
  (2013) 029} [\href{https://arxiv.org/abs/1308.4092}{{\ttfamily 1308.4092}}].

\bibitem{ILP}
A.~R. Its, O.~Lisovyy and A.~Prokhorov, \emph{{Monodromy dependence and
  connection formulae for isomonodromic tau functions}},
  \href{https://doi.org/10.1215/00127094-2017-0055}{\emph{Duke Mathematical
  Journal} {\bfseries 167} (2018) 1347 }.

\bibitem{ItLTy2014}
A.~Its, O.~Lisovyy and Y.~Tykhyy, \emph{{Connection problem for the
  sine-Gordon/Painlev\'e III tau function and irregular conformal blocks}},
  \href{https://arxiv.org/abs/1403.1235}{{\ttfamily 1403.1235}}.

\bibitem{Its2016}
A.~Its and A.~Prokhorov, \emph{Connection problem for the tau-function of the
  sine-gordon reduction of painlev{\'e}-iii equation via the riemann-hilbert
  approach}, {\emph{International Mathematics Research Notices} {\bfseries
  2016} (2016) 6856}.

\bibitem{Iwaki:2019zeq}
K.~Iwaki, \emph{{2-Parameter $\tau $-Function for the First Painlev\'e
  Equation: Topological Recursion and Direct Monodromy Problem via Exact WKB
  Analysis}}, \href{https://doi.org/10.1007/s00220-020-03769-2}{\emph{Commun.
  Math. Phys.} {\bfseries 377} (2020) 1047}
  [\href{https://arxiv.org/abs/1902.06439}{{\ttfamily 1902.06439}}].

\bibitem{Jeong:2020uxz}
S.~Jeong and N.~Nekrasov, \emph{{Riemann-Hilbert correspondence and blown up
  surface defects}}, \href{https://doi.org/10.1007/JHEP12(2020)006}{\emph{JHEP}
  {\bfseries 12} (2020) 006}
  [\href{https://arxiv.org/abs/2007.03660}{{\ttfamily 2007.03660}}].

\bibitem{JMU1}
M.~Jimbo, T.~Miwa and K.~Ueno, \emph{Monodromy preserving deformation of linear
  ordinary differential equations with rational coefficients: I. general theory
  and $\tau$-function},
  \href{https://doi.org/https://doi.org/10.1016/0167-2789(81)90013-0}{\emph{Physica
  D: Nonlinear Phenomena} {\bfseries 2} (1981) 306}.

\bibitem{KalKor}
C.~{Kalla} and D.~{Korotkin}, \emph{{Baker-Akhiezer Spinor Kernel and
  Tau-functions on Moduli Spaces of Meromorphic Differentials}},
  \href{https://doi.org/10.1007/s00220-014-2081-2}{\emph{Communications in
  Mathematical Physics} {\bfseries 331} (2014) 1191}
  [\href{https://arxiv.org/abs/1307.0481}{{\ttfamily 1307.0481}}].

\bibitem{kawai2005algebraic}
T.~Kawai and Y.~Takei, \emph{Algebraic analysis of singular perturbation
  theory}, vol.~227. American Mathematical Soc., 2005.

\bibitem{JDG}
A.~Kokotov and D.~Korotkin, \emph{{Tau-functions on spaces of Abelian
  differentials and higher genus generalizations of Ray-Singer formula}},
  \href{https://doi.org/10.4310/jdg/1242134368}{\emph{Journal of Differential
  Geometry} {\bfseries 82} (2009) 35 }.

\bibitem{Annalen}
D.~Korotkin, \emph{Solution of matrix riemann-hilbert problems with
  quasi-permutation monodromy matrices},
  \href{https://doi.org/10.1007/s00208-004-0528-z}{\emph{Mathematische Annalen}
  {\bfseries 329} (2004) 335}.

\bibitem{KorSam}
D.~Korotkin and H.~Samtleben, \emph{{Quantization of coset space sigma models
  coupled to two-dimensional gravity}},
  \href{https://doi.org/10.1007/s002200050247}{\emph{Commun. Math. Phys.}
  {\bfseries 190} (1997) 411}
  [\href{https://arxiv.org/abs/hep-th/9607095}{{\ttfamily hep-th/9607095}}].

\bibitem{KorRev}
D.~Korotkin, \emph{Bergman tau-function: From einstein equations and
  dubrovin-frobenius manifolds to geometry of moduli spaces}, {\emph{Integrable
  Systems and Algebraic Geometry: Volume 2} {\bfseries 459} (2020) 215}
  [\href{https://arxiv.org/abs/1812.03514}{{\ttfamily 1812.03514}}].

\bibitem{contemp}
D.~Korotkin and P.~Zograf, \emph{Tau function and the prym class},
  {\emph{Algebraic and geometric aspects of integrable systems and random
  matrices} (2013) 241}.

\bibitem{Nekrasov:2020qcq}
N.~Nekrasov, \emph{{Blowups in BPS/CFT correspondence, and Painlev\'e VI}},
  \href{https://arxiv.org/abs/2007.03646}{{\ttfamily 2007.03646}}.

\bibitem{Voros}
A.~Voros, \emph{The return of the quartic oscillator. the complex wkb method},
  {\emph{Annales de l'I.H.P. Physique th\'eorique} {\bfseries 39} (1983) 211}.

\end{thebibliography}

\end{document}